\documentclass[10pt,journal,compsoc]{IEEEtran}

\ifCLASSOPTIONcompsoc
  \usepackage[nocompress]{cite}
\else
  \usepackage{cite}
\fi

\ifCLASSINFOpdf
  \usepackage[pdftex]{graphicx}
  \graphicspath{{./figures/}}
\else

\fi

\usepackage{amsmath}

\usepackage{siunitx}
\usepackage{amsthm}
\usepackage{amssymb}
\usepackage{mathtools}

\usepackage{xspace}

\usepackage{algorithm}
\usepackage[noend]{algpseudocode}

\usepackage{color}
\usepackage{soul}

\usepackage{caption}
\usepackage{subcaption}

\usepackage{paralist}

\newtheorem{definition}{Definition}[section]
\newtheorem{lemma}[definition]{Lemma}
\newtheorem{theorem}[definition]{Theorem}
\newtheorem{corollary}[definition]{Corollary}

\theoremstyle{remark}
\newtheorem{example}[definition]{Example}
\newtheorem*{example*}{Example}
\newtheorem*{remark}{Remark}

\usepackage[capitalise]{cleveref}

\newcommand{\e}{\varepsilon}
\newcommand{\metas}{\textnormal{\texttt{M}}}
\newcommand{\mfu}{\textnormal{\texttt{M}}}

\newcommand{\IN}{\mathbb{N}}
\newcommand{\IR}{\mathbb{R}}

\DeclareMathAlphabet{\altmathcal}{OMS}{cmsy}{m}{n}

\DeclareMathOperator{\calG}{\altmathcal{G}}
\DeclareMathOperator{\calL}{\altmathcal{L}}
\DeclareMathOperator{\calO}{\altmathcal{O}}

\DeclareMathOperator{\md}{md}
\DeclareMathOperator{\clk}{clk}
\DeclareMathOperator{\ftone}{ft1}
\DeclareMathOperator{\fttwo}{ft2}

\newcommand{\OGCS}{\textsf{OffsetGCS}}
\newcommand{\CGCS}{\textsf{ClockedGCS}}
\newcommand{\FC}{\textsf{FC}}
\newcommand{\SC}{\textsf{SC}}
\newcommand{\FT}{\textsf{FT}}

\newcommand{\FCone}{\textsf{FC-1}}
\newcommand{\SCone}{\textsf{SC-1}}
\newcommand{\FTone}{\textsf{FT-1}}

\newcommand{\FCtwo}{\textsf{FC-2}}
\newcommand{\SCtwo}{\textsf{SC-2}}
\newcommand{\FTtwo}{\textsf{FT-2}}

\newcommand{\offset}{\widehat{O}}

\newcommand{\Tmax}{T_{\mathrm{max}}}
\newcommand{\Tmeas}{T_{\mathrm{meas}}}
\newcommand{\Tosc}{T_{\mathrm{osc}}}
\newcommand{\Tctr}{T_{\mathrm{ctr}}}
\newcommand{\Tclk}{T_{\mathrm{clk}}}
\newcommand{\modes}{\md}

\newcommand{\clkv}{\clk_v}
\newcommand{\clkw}{\clk_w}

\newcommand{\AND}{\textup{\texttt{and}}}
\newcommand{\OR}{\textup{\texttt{or}}}

\newcommand{\NOR}{\textup{\texttt{nor}}}

\newcommand{\abs}[1]{\left|#1\right|}

\newcommand{\ceil}[1]{\lceil{#1}\rceil}

\newcommand{\BFM}{\textsf{GCSoC}}

\algblockdefx[NAME]{On}{EndOn}%
    [1]{\textbf{at} #1 \textbf{do}}%
    {\textbf{end}}
\algtext*{EndOn}

\makeatletter
\DeclareRobustCommand\onedot{\futurelet\@let@token\@onedot}
\def\@onedot{\ifx\@let@token.\else.\null\fi\xspace}

\def\eg{e.g\onedot} 
\def\ie{i.e\onedot} 
\def\cf{cf\onedot} 
\def\wrt{w.r.t\onedot}
\def\resp{respectively\onedot}

\def\etal{et al\onedot}

\makeatother

\newcommand{\myparagraph}[1]{\medskip\noindent{\bf #1.}}

\begin{document}
\title{PALS: Distributed Gradient Clocking on Chip}
\author{Johannes Bund,
        Matthias Függer,
        Moti Medina%
\IEEEcompsocitemizethanks{%
 \IEEEcompsocthanksitem J.\ Bund was with the Engineering Faculty at Bar-Ilan University,
  Ramat Gan, Israel. Major parts where carried out when J.\ Bund was with CISPA Helmholtz
  Center for Information Security, Saarland, Germany.
  E-mail: bundjoh@biu.ac.il

  \IEEEcompsocthanksitem M.\ Függer was with CNRS, LMF, ENS Paris-Saclay, Université Paris-Saclay,
  91190 Gif-sur-Yvette, France.
  E-mail: mfuegger@lmf.cnrs.fr

  \IEEEcompsocthanksitem M.\ Medina was with the Engineering Faculty at Bar-Ilan University,
  Ramat Gan, Israel.
  E-mail: moti.medina@biu.ac.il

}\thanks{A conference version of this work appeared at IEEE ASYNC \cite{pals:async20}.}}
%
\markboth{}%
{}
\IEEEtitleabstractindextext{%
\begin{abstract}%
  Consider an arbitrary network of communicating modules on a chip, each requiring a local signal
  telling it when to execute a computational step.
There are three common solutions to generating such
  a local clock signal:
(i) by deriving it from a single, central clock source,
(ii) by local, free-running oscillators, or
(iii) by handshaking between neighboring modules.
Conceptually, each of these solutions is the result of a perceived dichotomy in which (sub)systems
  are either clocked or asynchronous.
We present a solution and its implementation that lies between these extremes.
Based on a distributed gradient clock synchronization algorithm, we show a novel design providing modules
  with local clocks, the frequency bounds of which are almost as good as those of free-running
  oscillators, yet neighboring modules are guaranteed to have a phase offset substantially smaller
  than one clock cycle.
Concretely, parameters obtained from a \SI{15}{\nano\meter} ASIC simulation running at \SI{2}{\giga\hertz}
yield mathematical worst-case bounds of \SI{20}{\pico\second} 
on the phase offset for a $32 \times 32$ node grid network.

\end{abstract}
}
%
\maketitle
\IEEEdisplaynontitleabstractindextext
\IEEEpeerreviewmaketitle
\IEEEraisesectionheading{\section{Introduction}\label{sec:intro}}

\noindent
Consider a circuit of dimension $W \times W$ consisting of
  comparably small modules of normalized dimension~$1$ that
  predominately communicate with physically close modules.

\myparagraph{Local clock signals}
Each circuit module requires a local clock signal to trigger its
  computational steps.
There are two extreme approaches to provide these clock signals:
(i) In the synchronous approach, a clock signal is distributed via a clock tree.
The clock's period is chosen to ensure lock-step computational rounds between all modules and, thus, in particular, for communicating modules.
(ii) In the asynchronous approach, modules, potentially as small as a single gate, generate the clock signal locally via handshaking with all communication partners.

From the modules' perspective, both methodologies provide local clock ticks with different guarantees.
We focus on three measures:
(a) Whether they ensure a \emph{round structure}, i.e.,
  there is a time-independent tick offset~$R$,
  such that data that is provided at tick~$k$ is guaranteed to be available at the receiving module at tick~$k+R$.
(b) The \emph{local skew}, i.e., the maximum difference in
  time between any two $k^\text{th}$ ticks at modules that communicate with each other.
(c) The \emph{waiting time}, i.e., the maximal time between two successive ticks.
Note that (b+c) combined with a minimal time between successive ticks allows one to construct a round structure.

\myparagraph{Limits of the extremes}
Ideally, one would wish for a circuit with a round structure, a small local skew, and a small waiting time.
While fully synchronous and asynchronous circuits provide a round structure,
  they have significant local skew or waiting times.
The \emph{local skew} in synchronous circuits has been shown to \emph{grow linearly} with the circuit
  width~$W$; see Section~\ref{subsec:clocktree}.
On the other hand, causal acknowledge chains in asynchronous systems can span the entire system,
  resulting in \emph{waiting times} that \emph{grow linearly} with~$W$.

Fully synchronous or asynchronous systems are indeed
  rare in practice: in current synchronous systems, there are numerous clock domains with
asynchronous interfaces in between~\cite{foster2015trends}.
Full asynchronous, delay-insensitive
  circuits~\cite{martin1986compiling} suffer from substantial computational
  limitations~\cite{martin1990limitations,manohar2017eventual,manohar2019asynchronous}
  and provide no timing guarantees, rendering them unsuitable for many applications.
Accordingly, most real-world asynchronous systems
  will utilize timing assumptions on some components.

\myparagraph{A systematic tradeoff: GALS}
Globally Asynchronous Locally Synchronous (GALS)
  systems~\cite{chapiro1984globally,teehan2007survey}
  are a systematic approach between both extremes.
Unlike a single synchronous region (clocked circuits), and no clocked regions (delay-insensitive), GALS systems have several clock islands communicating asynchronously via handshakes.
If the width $W$ of the synchronous islands is small, the clock islands can maintain small skew locally.
However, the gain comes at the expense of \emph{no round structure}: Clocks on different islands may drift apart arbitrarily.
Furthermore, communication across clock domains requires passing through synchronizers~\cite{teehan2007survey,dobkin2004data}.
Besides the disadvantage of a \emph{non-zero probability to cause metastable upsets},
  the synchronizers incur $2$ or more clock cycles of
  additional \emph{communication latency} if they are in the data path.

Alternative solutions without synchronizers in the data path
  have been proposed in \cite{DDX95low,CG03efficient}.
The designs either skip clock cycles or switch to a clock
  signal shifted by half a period when the transmitter and receiver clock risk violating setup/hold conditions.
The signal that indicates this choice (skip/switch) is synchronized without additional latency to the data path.
Depending on the implementation and intended guarantees, the additional latency is in the order of a clock period.
While this can, in principle, be brought down to the order of setup/hold-windows, such designs would require considerable logical overhead and fine-tuning of delays.
An application-level transmission may be delayed by such a time slot.
In \cite{DDX95low}, this additional delay can be up to $2$ clock periods when a
  so-called no-data packet is over-sampled.
Further, there is a \emph{non-zero probability of metastable upsets}, and applying such a scheme has to insert no-data packets periodically.

Finally, consider a potential application that runs on top of such schemes and uses handshaking to make sure all its packets of a (logical) time
  step have arrived before the next time step is locally initiated. It faces
  the same problem as a fully asynchronous design: the
  \emph{waiting time grows linearly} with the circuit dimension.

\myparagraph{Solutions with round structure}
A fully synchronous system provides two convenient properties:
(i) metastability-free communication, and
(ii) a round structure, that is, no need for handshaking between the synchronous islands of a GALS system.
Abandoning the synchronous structure leads to the loss of these properties.

Solutions that directly provide a round structure have been proposed.
Examples are GALS architectures with pausible   clocks~\cite{yun1996pausible,fan2009analysis}, distributed clock generation
  algorithms like DARTS~\cite{FS12:DC} and FATAL~\cite{dolev2014fault},
  wave clock distribution~\cite{fischer20239}, and
  distributed clocking grids~\cite{fairbanks2005self}.
However, pausible clocks suffer from potentially unbounded waiting times
  due to metastability, and DARTS and FATAL require essentially fully connected
  communication networks. 
  The solution in~\cite{fischer20239} employs analog cross-coupling of clock buffers to distribute a single clock source over a grid network.
Here, we show how to design a purely digital system that utilizes an algorithmic approach to synchronize many
  clock domains. 
A digital version of the Fairbanks clock generation grid~\cite{fairbanks2005self} is analyzed in this work
  and shown to lead to linear waiting times; see Section~\ref{sec:fairbanks}. 

\myparagraph{The PALS approach}
In this work, we present a different approach that combines
  a \emph{round structure} with \emph{low local skew}, a
  \emph{low waiting time}, and provable \emph{absence of metastable-upsets}.
The PALS approach can be regarded as a drop-in replacement for GALS systems; it uses locally synchronous islands and global communication between islands. 
It improves over GALS systems by adding a round structure.
  
Our design is based on the \emph{distributed gradient clock
  synchronization (\textsf{GCS}) algorithm} by Lenzen et~al.~\cite{lenzen10tight}, in which the goal
  is to minimize the worst-case clock skew between adjacent nodes in a network.
In our setting, the modules correspond to nodes; an edge connects them if they directly communicate (i.e., exchange data).
More precisely, let $D$ be the diameter of the network and
  $\rho$ be the (unintended) drift of the clock of a clocked region,
  $\mu > 2 \rho$ a freely chosen constant, and $\delta$ an upper bound on how
  precisely the phase difference to neighboring clocked regions is known.
Then:
\begin{itemize}
  \item The synchronized clocks are guaranteed to run at normalized rates
  between $1$ and $(1+\mu)(1+\rho)$, i.e., have \emph{constant waiting time}.
  \item The \emph{local skew} is bounded by $\calO(\delta \log_{\mu/\rho}D)$.
  \item The \emph{global skew}, i.e., the maximum phase offset between any two
  nodes in the system, is $\calO(\delta D)$.
\end{itemize}
In other words, the synchronized clocks are almost as good as free-running
  clocks in a GALS system, with drift $\rho$,
  yet the \emph{local skew} grows only
  \emph{logarithmically in the chip width~$W$}.

We extend the conference version~\cite{pals:async20} by
  an in-depth analysis of the PALS algorithm and additional simulations demonstrating its performance.
We further add a comparison to a state-of-the-art clock generation grid~\cite{fairbanks2005self}.

\myparagraph{Outline and results}
After the introduction in this section, we discuss the computational model in Section~\ref{sec:model}.
In Section~\ref{sec:algs}, we briefly present the \textsf{GCS} algorithm before discussing a variant of the algorithm (called $\OGCS$) used in this work.
We break down the $\OGCS$ into hardware modules and specify these modules in
  Section~\ref{sec:hardware}.
The algorithm carried out by the hardware modules is denoted by $\CGCS$.
Our main theorem (Theorem~\ref{lem:implements}) states that every hardware system
  that implements $\CGCS$ maintains the skew bounds of the \textsf{GCS} algorithm.
An implementation of the hardware modules on register-transfer-level, which we denote by
  \BFM, is discussed in Section~\ref{sec:impl}.
We conclude this work with simulations of this implementation in Section~\ref{sec:simu}.
Implementation and simulation are carried out in the $\SI{15}{\nano\meter}$ FinFET-based NanGate
  OCL~\cite{martins2015open}.
For $\SI{2}{\giga\hertz}$ clock sources with an assumed drift of $\rho=10^{-5}$, and $\mu =
  10^{-3}$, our simple sample implementation guarantees that $\delta\leq
  \SI{5}{\pico\second}$ in the worst case.
The resulting local skew is $\SI{20}{\pico\second}$, 
  which is well below a clock cycle.
We stress that this enables much faster communication than
  handshake-based solutions, which incur synchronizer delay.
We conclude with a comparison of the performance of our solution by SPICE simulations to a
  digital version of a Fairbanks grid.

\section{Computational Model}\label{sec:model}

\myparagraph{Network, Communication, and Timing}
The network of communicating synchronous PALS islands is
  modeled by an undirected graph $G=(V,E)$, where the set of nodes $V$ is the set of islands and there is an edge $(v,w) \in E$, if $v$ and $w$ communicate.
Edges are bidirectional, i.e., edges $(v,w)$ and $(w,v)$ are the same edge.
Furthermore, for each $v\in V$, $E$ contains edge $(v,v)$.
The diameter $D$ of a network is the maximum distance over all pairs of nodes, where the distance between two nodes is the length of a shortest
  path connecting those nodes in the network.

We denote real time, \ie, an external reference time for analysis,
  by \emph{Newtonian} time.
A node has no access to Newtonian time,
  but has its own (internal) time reference.

Nodes communicate by sending content-less messages, known as \emph{pulses}.
A pulse is sent via broadcast to all neighboring nodes.
The message delay is the time a pulse travels between the sender and receiver.
It is constrained by a maximum delay $d$ and a minimum delay $d-U$, where
  $U$ is the delay uncertainty.
A pulse sent by a node at Newtonian time $t$ is received between time $t+d-U$
  and time $t+d$.

\myparagraph{Hardware Clock}
Each node can locally measure the progress of time.
For example, a node may do this via a local ring oscillator.
For the analysis, we abstract any such device by the mathematical concept of a
  \emph{hardware clock}.
A hardware clock is prone to uncertainty, which we model by a variable rate
  that may change over time.
The uncertainty is called the \emph{hardware clock drift} (short \emph{clock drift}).
Formally, for each node $v\in V$ there is an
  integrable function $h_v\colon\IR_{\geq 0}\rightarrow\IR$
  called the hardware clock rate.
Parameter $\rho > 0$ is an upper bound on the one-sided hardware clock drift of all nodes.
The hardware clock rate satisfies $1\leq h_v(t)\leq 1+\rho$ for all $t\in\IR_{\geq 0}$.
The hardware clock value of $v$ at time $t$, $H_v(t)$, is then defined by
\begin{equation*}
  H_v(t)=\int_{0}^t h_v(\tau)d\tau+H_v(0)\,,
\end{equation*}
where $H_v(0)$ is the initial value of $v$'s hardware clock at Newtonian time~$0$.
A node that has measured its hardware clock to advance by $T$ knows
  that the real time difference is in $[T/(1+\rho), T]$.

\myparagraph{Logical Clock}
While a hardware clock allows a node to measure time differences, its rate cannot
  be controlled.
The logical clock is a hardware clock that can also be controlled.
Indeed, we will use an adjustable ring oscillator in this work as a
  logical clock.
Formally, the local clock signal a node produces is given by
  $L_v\colon\IR_{\geq 0}\rightarrow\IR$,
  where $L_v(t)$ is the phase (normalized by $2\pi$) of the clock signal since the
  first clock tick.\footnote{For
  example, a (perfect, non-drifting) logical clock with
  frequency $f$ has phase $2\pi \cdot f \cdot t$
  at time~$t$, and normalized (by $2\pi$) phase $f\cdot t$.
The logical clock thus advances by a full step of one every $1/f$ time
  but continuously advances in between.}
A node's logical clock is initialized to $H_v(0)$ and
  follows the hardware clock's rate but is adjustable by a constant factor.
In our algorithm, the logical clock will be implemented by the node's local adjustable ring
  oscillator that has only two modes: slow and fast.
Their rate differs by factor $\mu$, slow has (normalized) rate $1$ and
  fast has $(1+\mu)$.

\myparagraph{Skew}
The skew between two nodes describes the difference in their logical clock values.
The upper bound on the skew is a figure of merit for clock synchronization algorithms.
We regard two types of skew in a system, the \emph{global skew} and the \emph{local skew}.
\begin{definition}[global and local skew]
  The \emph{global skew} $\calG(t)$ is the maximum skew between any two nodes
  in the network. Formally, it is defined by
  \begin{equation*}
    \calG(t)\coloneqq\max_{v,w\in V}\left\{L_w(t)-L_v(t)\right\}\,.
  \end{equation*}
  The \emph{local skew} $\calL(t)$ is the maximum skew between any two neighboring nodes
  in the network. Formally, it is defined by
  \begin{equation*}
    \calL(t)\coloneqq\max_{(v,w)\in E}\left\{L_w(t)-L_v(t)\right\}\,.
  \end{equation*}
\end{definition}

\myparagraph{Our Goal}
Given a network of nodes and the nodes' local oscillator parameters,
  the goal is to provide an algorithm (i.e., a circuit)
  that controls the slow and fast clock speed signals at each node such that small,
  bounded, local, and global skews are ensured.

\section{Algorithm and Skew Bound Guarantees}\label{sec:algs}

\subsection{Gradient Clock Synchronization}\label{sec:gcs}

We start by recalling the class of \textsf{GCS} algorithms studied by
  Lenzen \etal~\cite{lenzen10tight}.

Intuitively, a \textsf{GCS} algorithm executed by node $v$ continuously measures the
  skew to each neighbor $w$.
By a set of rules, the algorithm decides whether to progress the logical clock
  at a fast or a slow rate.
Lenzen \etal showed that such \textsf{GCS} algorithms achieve close synchronization between
  neighboring nodes in an
  arbitrary network, \ie, minimize $\calL(t)$.
Let $\delta$ be an upper bound on how precisely the skew between neighbors is known.
Provided that the global skew does not exceed a bound of $\calO(\delta D)$, \textsf{GCS} achieves asymptotically optimal local skew bounded by $\calO(\delta \log_{\mu/\rho}D)$.
In other words, local skew grows only logarithmically in the hop diameter of the network; still, we have clocks that progress at a minimum rate of $1$.
The local and global skew bounds are asymptotically optimal~\cite{lenzen10tight}.

\subsection{GCS and OffsetGCS Algorithm}\label{sec:prelim-gcs-alg}

The \textsf{GCS} algorithm by Lenzen \etal computes a logical clock from the hardware clock in two
  different modes, fast and slow.
In slow mode, the logical clock follows the rate of the hardware clock.
In fast mode, the logical clock advances at the hardware clock rate and \emph{speedup factor} $\mu>0$.
Formally, a node in fast mode advances its logical clock with rate $(1+\mu)h_v(t)$, where
  $\mu$ is chosen by the designer.
A node controls its binary mode signal $\gamma_v(t)\in\{0,1\}$ to adjust its logical clock.
In \emph{fast mode} $\gamma_v$ is set to $1$ and,
  accordingly, in \emph{slow mode} $\gamma_v$ is set to $0$.
The logical clock value of $v$ at time $t$ with initial value $H_v(0)$ thus is
\begin{equation*}
  L_v(t)=\int_0^t (1+\mu\cdot\gamma_v(\tau))\,h_v(\tau)\,d\tau+H_v(0)\,.
\end{equation*}
A node in fast mode must be able to catch up to a node in slow mode.
Hence, we pose the constraint that fast mode (without clock drift) can never be slower
  than slow mode (with clock drift).
This can be formalized as
\begin{equation*}
  1+\rho<1+\mu\,.\label{eqn:murho}
\end{equation*}

\noindent
The algorithm specifies two conditions that control when to switch between fast and slow modes.
Accordingly, conditions are named \emph{fast condition} ($\FC$) and \emph{slow condition} ($\SC$).
The algorithm is parameterized by $\kappa$, which determines the synchronization quality.

\begin{definition}[fast and slow condition]\label{def:gcsconditions}
  Let $\kappa\in\IR^+$ be a positive, non-zero, real number.
  A node $v\in V$ satisfies the \emph{fast condition at time $t$} if there is a natural number $s\in \IN = \{0,1,\dots\}$ such that both:
  \begin{subequations}
    \begin{align}
      \exists (v,x)\in E : L_x(t)-L_v(t)&\geq (2s+1)\kappa     \label{eqn:fc1}\tag{\FCone}\\
      \forall (v,y)\in E : L_y(t)-L_v(t)&\geq -(2s+1)\kappa      \label{eqn:fc2}\tag{\FCtwo}
    \end{align}
  \end{subequations}
  Node $v\in V$ satisfies the \emph{slow condition at time $t$} if there is a natural number $s\in \IN$ such that the following conditions hold:
  \begin{subequations}
    \begin{align}
      \exists (v,x)\in E : L_x(t)-L_v(t)&\leq -2s\kappa     \label{eqn:sc1}\tag{\SCone}\\
      \forall (v,y)\in E : L_y(t)-L_v(t)&\leq 2s\kappa      \label{eqn:sc2}\tag{\SCtwo}
    \end{align}
  \end{subequations}
\end{definition}

\noindent
Node $v$ satisfies the fast condition if there is at least one node $u$, that is ahead of
$v$ and no other node behind
$v$ exceeds the absolute skew between $v$ and $u$.
The slow condition is satisfied if there is a node $u$ behind $v$ that has a larger absolute
skew to $v$ than all nodes ahead of $v$. The thresholds use
odd multiples of $\kappa$ for the fast condition and even multiples of $\kappa$
for the slow condition to ensure mutual exclusion.

If $v$ is the node with the largest logical clock value in the network, then all other
nodes are behind $v$, they have negative skew.
Thus, the slow condition is satisfied for $s=0$.
Accordingly, if $v$ is the node with the smallest clock value in the network, then it satisfies
the fast condition as all skews to other nodes are positive.

\begin{definition}
  \label{dfn:gcs-implementation}\label{def:gcsimpl}
  An algorithm is a \emph{\textsf{GCS} algorithm with parameters $\rho, \mu, \kappa$}
  if the following invariants hold, for every node $v \in V$ and all times $t, t'$:
  \begin{align*}
    & \mu > \rho\,\tag{I1}\label{eqn:i1}\\
    & L_v(t') - L_v(t)\! \in\! [1,1 + \mu] \cdot (H_v(t') - H_v(t))\,\tag{I2}\label{eqn:i2}\\
    & \text{if $v$ satisfies $\FC$ at time $t$ then $v$ is in fast mode at time $t$}\tag{I3}\label{eqn:i3}\\
    & \text{if $v$ satisfies $\SC$ at time $t$ then $v$ is in slow mode at time $t$}\tag{I4}\label{eqn:i4}\\
  \end{align*}
\end{definition}

\noindent
Invariant \eqref{eqn:i2} states that the rate of the logical clock is at least
  the rate of the hardware clock and at most $(1 + \mu)$ times the rate of the hardware clock.

\begin{remark}
  Every algorithm that meets \cref{def:gcsimpl} is a \textsf{GCS} algorithm.
  In this work we focus on the algorithm by Lenzen, Locher, and Wattenhofer \cite{lenzen10tight},
    which we state in Algorithm~\ref{alg:gcs}.
\end{remark}

\myparagraph{Maximal and Minimal Offsets}
The conditions in \cref{def:gcsconditions} can be reformulated using the maximal
  and minimal offset.
Maximal and minimal offsets at node $v$ are given by
\begin{align}
  O_{\max}(t)&\coloneqq\max_{(v,x)\in E}\{L_x(t)-L_v(t)\}\,,\\
  O_{\min}(t)&\coloneqq\min_{(v,x)\in E}\{L_x(t)-L_v(t)\}\,.
\end{align}
The node with the largest offset to $v$ is the node that is  ahead of $v$ the most, and the node with the smallest offset to $v$ is the node most behind $v$.

A node $v\in V$ satisfies the fast condition if some neighbor reached a (positive)
  threshold and no other neighbor crossed the corresponding negative offset.
If the $O_{\max}(t)$ reached a certain threshold we are certain that some
  neighbor reached this offset.
Accordingly if $O_{\min}(t)$ is larger than the corresponding negative offset,
  no neighbor crossed the corresponding negative offset.
Formally, we replace \cref{eqn:fc1,eqn:fc2} and \cref{eqn:sc1,eqn:sc2} in 
  \cref{def:gcsconditions} by
\begin{align}
  &O_{\max}(t)\geq(2s+1)\kappa\,,  \label{eqn:restatefc1}\tag{\FCone} \\
  &O_{\min}(t)\,\geq-(2s+1)\kappa\,,  \label{eqn:restatefc2}\tag{\FCtwo} \\
  &O_{\min}(t)\,\leq-2s\kappa\,,\text{and} \label{eqn:restatesc1}\tag{\SCone} \\
  &O_{\max}(t)\leq2s\kappa\,.  \label{eqn:restatesc2}\tag{\SCtwo}
\end{align}

\myparagraph{Offset Estimates}
Nodes have no access to logical clocks of their neighbors.
Hence, precise skews remain unknown to the node. In order to fulfill the invariants of the
algorithm a node maintains an estimate of each offset to a neighbor. Offset and skew
are the same, we use the terms interchangeably. The
\emph{offset estimate} of node $v$ to its neighbor $w$ is denoted by
$\widehat{O}_w$.
Intuitively, we have $\widehat{O}_w(t)\approx L_w(t)-L_v(t)$.
Parameter $\delta$ gives a two-sided bound on the estimates.
\begin{equation}
  \left|\widehat{O}_w(t) - (L_w(t)-L_v(t))\right| \leq \delta \label{eqn:delta-pm}
\end{equation}

\noindent
Given an estimate of each neighboring clock, the \textsf{GCS} algorithm
  specifies the \emph{fast trigger} ($\FT$).
Each node determines by $\FT$ whether to go fast or slow.
A node that satisfies $\FC$ must satisfy $\FT$, but a node that satisfies
  $\SC$ must not satisfy $\FT$.

\begin{definition}[fast trigger]\label{def:fasttrigger}
  Let $\kappa\in\IR^+$ be a positive, non-zero, real number.
  A node $v\in V$ satisfies the \emph{fast trigger at time~$t$} if there is a natural number $s\in \IN$
  such that both:
  \begin{subequations}
    \begin{align}
      \exists (v,x)\in E : \widehat{O}_x(t)&\geq(2s+1)\kappa-\delta  \label{eqn:ft1}\tag{\FTone} \\
      \forall (v,y)\in E : \widehat{O}_y(t)&\geq-(2s+1)\kappa-\delta  \label{eqn:ft2}\tag{\FTtwo}
    \end{align}
  \end{subequations}
\end{definition}

\noindent
We are now able to state the \textsf{GCS} algorithm in \cref{alg:gcs}.
Intuitively, the \textsf{GCS} algorithm checks the $\FT$ at all times. If $v$ satisfies $\FT$
then $v$ switches to fast mode, otherwise $v$ defaults to slow mode.

\begin{algorithm}[t]
  \begin{algorithmic}[1]
    \On{each time $t$}
      \For{each neighbor $w$}
        \State $o_w\gets\widehat{O}_w(t)$ \Comment{save offset estimate to $w$}
      \EndFor
      \State $\ftone_s \gets \exists w : o_w\geq(2s+1)\kappa-\delta$
      \State $\fttwo_s \gets \forall w : o_w\geq-(2s+1)\kappa-\delta$
      \If{ $\exists s : \ftone_s \land \fttwo_s$ }
        \State $\gamma_v(t)\gets 1$ \Comment{switch to fast mode}
      \Else
        \State $\gamma_v(t)\gets 0$ \Comment{switch to slow mode}
      \EndIf
    \EndOn
  \end{algorithmic}
  \caption{\textsf{GCS} algorithm at node $v$, where $\ftone_s$, $\fttwo_s$, and $o_w$ are variables.}
  \label{alg:gcs}
\end{algorithm}

\begin{remark}
As the decision to run fast or slow is a discrete decision,
  a circuit implementation will be prone to metastability~\cite{marino:general}.
We focus on this problem in \cref{par:clockedalg}.
\end{remark}

\myparagraph{Maximal and Minimal Offset Estimates}
The fast trigger of Lenzen \etal can be redefined using the largest
and smallest offset estimate of node $v$.
We define the maximal and the minimal estimate of $v$'s offset estimates by
\begin{align*}
  \widehat{O}_{\max} \coloneqq \max_{(v,x)\in E}\{\widehat{O}_x\}\,,\text{ and }
  \widehat{O}_{\min} \coloneqq \min_{(v,x)\in E}\{\widehat{O}_x\}\,.
\end{align*}
Then, we replace \cref{eqn:ft1,eqn:ft2} in \cref{def:fasttrigger} by
\begin{align}
  \widehat{O}_{\max}(t)&\geq(2s+1)\kappa-\delta\,,  \label{eqn:restateft1}\tag{\FTone} \\
  \widehat{O}_{\min}(t)&\geq-(2s+1)\kappa-\delta\,.  \label{eqn:restateft2}\tag{\FTtwo}
\end{align}
We are thus able to restate the \textsf{GCS} algorithm (Algorithm~\ref{alg:gcs})
  in Algorithm~\ref{alg:ogcs}.
Rather than quantifiers ``exists'' and ``for all''
  over all outgoing edges we make use of $\widehat{O}_{\max}$
  and $\widehat{O}_{\min}$.
The algorithm, referred to as $\OGCS$, is stated in \cref{alg:ogcs}.

Next, we show that in $\OGCS$ we can also bound the number of thresholds,
i.e., we bound the maximum value of $s$ in \cref{alg:gcs}, line $6$.
  
\myparagraph{Bound on $s$}
In case of a bounded local skew, which we will later show to be the case, we can
  bound the maximal number of steps $s$ that need to be measured.
Let $\calL \coloneqq \sup_t\calL(t)$ be the largest skew between two neighbors.
Let $\ell$ be the largest number such that $(2\ell+1)\kappa-\delta \leq \calL$.
Then, node $v\in V$ satisfies the fast trigger at time $t$ if there is an $s\in [\ell+1]$
  such that \cref{eqn:restateft1,eqn:restateft2} hold.

\begin{algorithm}[t]
  \begin{algorithmic}[1]
    \On{each time $t$}
      \For{each adjacent node $w$}
        \State $o_w\gets\widehat{O}_w(t)$ \Comment{save offset estimate to $w$}
      \EndFor
      \State $o_{\max} \gets \max_{(v,w)\in E}\{o_w\}$ \Comment{compute and save max}
      \State $o_{\min} \gets \min_{(v,w)\in E}\{o_w\}$ \Comment{and min estimates}
      \For{$s\in[\ell+1]$}
        \State $\ftone_s \gets o_{\max} \geq (2s+1)\kappa-\delta$
        \State $\fttwo_s \gets o_{\min} \geq-(2s+1)\kappa-\delta$
      \EndFor
      \If{ $\exists s \in[\ell+1]: \ftone_s \land \fttwo_s$ }
        \State $\gamma_v(t)\gets 1$ \Comment{switch to fast mode}
      \Else
        \State $\gamma_v(t)\gets 0$ \Comment{switch to slow mode}
      \EndIf
    \EndOn
  \end{algorithmic}
  \caption{$\OGCS$ algorithm at node $v$, where $o_w$, $o_{\max}$, $o_{\min}$,
  $\ftone_s$, and $\fttwo_s$ are variables that store a value.}
  \label{alg:ogcs}
\end{algorithm}

\myparagraph{Visualization}
In \cref{fig:gcs_large} we depict the conditions of $\OGCS$.
Along the axis we mark the offsets, where the $x$-axis marks the maximal and the
  $y$-axis marks the minimal offset.
Conditions $\FC$, $\SC$ and $\FT$ can be marked as areas.
The fast condition, defined by \cref{eqn:restatefc1,eqn:restatefc2}, is
  marked in yellow.
The slow condition, defined by \cref{eqn:restatesc1,eqn:restatesc2}, is
  marked in blue.
The fast trigger, defined by \cref{eqn:restateft1,eqn:restateft2}, is
  a translation of $\FC$ by $\delta$ to the left and $\delta$ down.
It is marked by the orange are (including the yellow area).
We require that $\kappa > 2\delta$ (see \cref{lem:uncertainty}),
  thus, the gap between $\FC$ and $\SC$ is larger than $2\delta$.

Fix a node $v$, we mark maximal and minimal offsets as a point $(O_{\min}, O_{\max})$.
In the \textsf{GCS} algorithm $v$ goes to fast mode at any time where $(O_{\min}, O_{\max})$
  falls into the $\FC$ (yellow) region.
Similarly, if $(O_{\min}, O_{\max})$ falls into the $\SC$ (blue) region, $v$
  goes to slow mode.
In between both regions $v$ is free to choose any speed between fast and slow mode.
The offset estimates of $v$ are given by point $(\widehat{O}_{\min}, \widehat{O}_{\max})$,
  which we mark as a cross.
Due to \cref{eqn:delta-pm} the cross may fall into the $\delta$ surrounding of
  $(O_{\min}, O_{\max})$.
We mark this by a larger box surrounding $(O_{\min}, O_{\max})$.

A node that executes $\OGCS$ chooses to go to fast or slow mode depending on whether
  $(\widehat{O}_{\min}, \widehat{O}_{\max})$ falls into the $\FT$ (orange and yellow)
  region.
From the visualization one can see that for any $(O_{\min}, O_{\max})$ in the $\FC$ region,
  $(\widehat{O}_{\min}, \widehat{O}_{\max})$ will be
  within the $\FT$ region, and can never cause $\OGCS$ to go slow.
Similarly any $(O_{\min}, O_{\max})$ in the $\SC$ region can never cause $\OGCS$
  to go fast.

\begin{figure}
  \centering
  \includegraphics[width=0.8\linewidth]{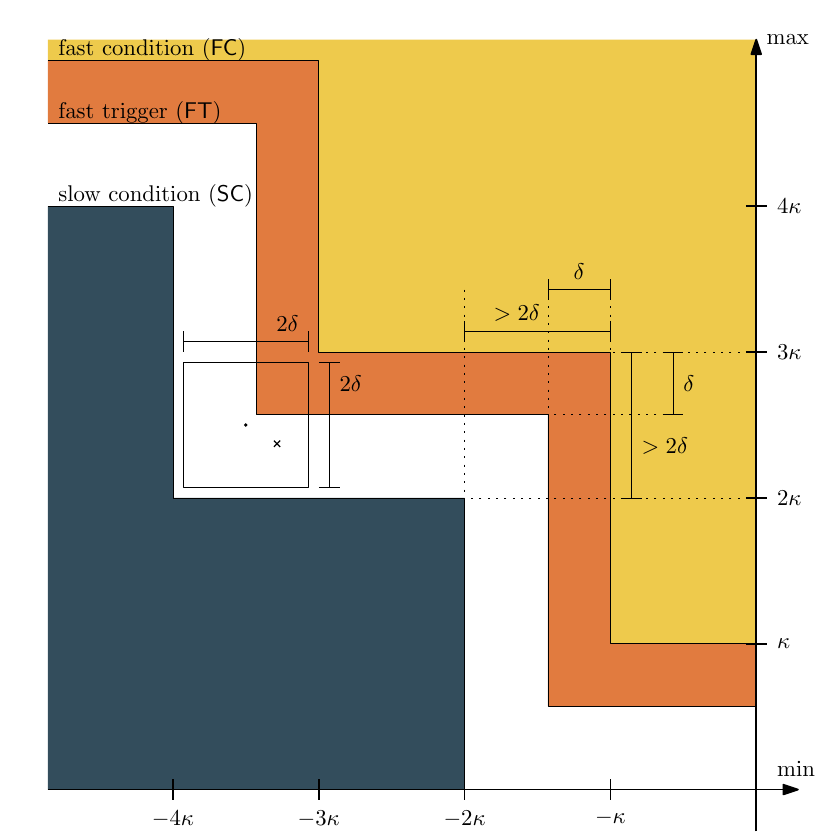}
  \caption{Visualization of $\FC$, $\SC$ and $\FT$ in the plane of $\min$ and
    $\max$ of the offsets.
    The exact measurement $(O_{\min}, O_{\max})$ is denoted by a point.
    The actual measurement $(\widehat{O}_{\min}, \widehat{O}_{\max})$ is depicted
    as a small cross.
    We denote the maximum measurement error by a square surrounding the exact measurement.
    The $\OGCS$ algorithm switches to fast when $(\widehat{O}_{\min}, \widehat{O}_{\max})$
    is within the $\FT$ region and to slow otherwise.}
  \label{fig:gcs_large}
\end{figure}

\subsection{Analysis of the $\OGCS$ Algorithm}\label{sec:gcs-analysis}

In what follows, we prove that for a suitable choice of parameters 
  (i.e., $\rho$, $\mu$, $\kappa$, $\delta$, and $\ell$), $\OGCS$ is a
  \textsf{GCS} algorithm as defined in \cref{dfn:gcs-implementation}.
Theorem~\ref{thm:gcs} implies that $\OGCS$ maintains tight skew bounds.
For the proof of Theorem~\ref{thm:gcs}, we refer the reader to \cite{pals:arxiv}.

\begin{theorem}
\label{thm:gcs}
Suppose algorithm $A$ is a \textsf{GCS} algorithm according to \cref{def:gcsimpl}
  with $\mu > 2 \rho$.
Then, $A$ maintains global and local skew of
\begin{align*}
  &\calG(t) \leq \frac{\mu \kappa D}{\mu - 2 \rho} &
  \calL(t) \leq \left(\left\lceil\log_{\mu / \rho} \frac{\mu D}{\mu - 2 \rho}\right\rceil + 1 \right)\kappa
\end{align*}
either 1) for all $t>0$ if $\calL(0)\leq\kappa$ or
2) for sufficiently large $t\geq T$, where $T\in\calO((\calG(0)+\kappa D)/(\mu-2\rho))$.
\end{theorem}

\myparagraph{Uncertainty sources in $\OGCS$ implementations}
In our analysis, we distinguish between two sources of uncertainty:
(i) The \emph{propagation delay uncertainty} $\delta_0$.
This is the absolute timing variation added to the measurement error due to propagation delays, e.g.,
  wire and gate delays on the path from the clock source to the measurement module.
(ii) The measurement error resulting from unknown clock rates.
Denote by $T_{\max}$ the \emph{time between initiating a measurement and using it}
  to control the logical clock speed.
During this time, logical clocks advance at rates that are not precisely known.
This adds to the measurement error because the actual difference might
  increase or decrease compared to the measured difference.

We denote an upper bound on the combined error by $\delta$; the relation of
  $\delta$ to $\delta_0$ and $T_{\max}$ is elaborated in \cref{sec:cgcs}.
Given $\delta$, we seek to choose $\kappa$ as small as possible to obtain a
  small local skew bound
  according to Theorem~\ref{thm:gcs}.
In the following, we show constraints on parameters like $\kappa$
  such that an instance of $\OGCS$, we will refer to it
  as a particular implementation of $\OGCS$, is a \textsf{GCS} algorithm.

\begin{theorem}
  \label{lem:uncertainty}
  An implementation of $\OGCS$ is a \textsf{GCS} algorithm if for all times~$t$ it satisfies
	\begin{inparaenum}[(i)]
    \item $\mu > 2\rho$
    \item $\left|\widehat{O}_w(t) - (L_w(t)-L_v(t))\right| \leq \delta$ 
    \item $\kappa > 2 \delta$
    \item $\sup\{s\in\IN | (2s+1)\kappa \leq \calL + 2\delta \} \leq \ell < \infty$
  \end{inparaenum}
\end{theorem}
\begin{proof}
  To show the statement, we verify the conditions of
    \cref{dfn:gcs-implementation}.
  By assumption condition \eqref{eqn:i1} is satisfied.
  Condition \eqref{eqn:i2} is a direct consequence of the algorithm specification.
  For Condition \eqref{eqn:i3}, suppose first that~$v$ satisfies the fast
  condition at time $t$.
  There exists some $s \in \IN$ and neighbor~$x$ of~$v$ such that
  $L_x(t) - L_v(t) \geq (2 s + 1) \kappa$.
  Therefore, by \eqref{eqn:delta-pm}, $\offset_x(t) \geq (2 s + 1) \kappa - \delta$,
    so that \eqref{eqn:ft1} is satisfied.
  Further, from the definitions of the local skew and~$\delta$,
    $\calL + \delta \geq \offset_x(t)$.
  Combining the above inequalities yields
  $\calL + \delta \geq (2 s + 1) \kappa - \delta$.
  By assumption, $s$ fulfills $s \leq \ell < \infty$.
  Thus, $\ftone_s$ is set to true.

  Similarly, since~$v$ satisfies the fast condition, all of its neighbors $y$
  satisfy $L_v(t) - L_y(t) \leq (2 s + 1) \kappa \leq \calL$.
  Therefore, $\offset_y(t) \geq -(2 s + 1) \kappa -\delta \geq -\calL-\delta$,
    hence \eqref{eqn:ft2} is satisfied for the same value of $s$.
  From $-(2 s + 1) \kappa -\delta \geq -\calL-\delta$, it is
    $(2 s + 1) \kappa \leq \calL$.
  Thus, $s \leq \ell < \infty$ and $\fttwo_s$ is set to true.
  Consequently, $v$ runs in fast mode at time $t$.

  It remains to show that if $v$ satisfies $\SC$ at time $t$,
  it does not satisfy $\FT$ at time $t$ and is in slow mode.
  Suppose, for contradiction, that $v$ satisfies $\SC$ and $\FT$ at time~$t$.
  Fix $s\in\IN$ such that $\SC$ is satisfied and $s'\in\IN$ such that $\FT$
    is satisfied.
  Then, by \eqref{eqn:delta-pm} and $\SCtwo$,
  \begin{align*}
    \offset_{\max}(t) \leq 2s\kappa+\delta\,.
  \end{align*}
  Thus, by $\FTone$,
    $(2 s' + 1) \kappa - \delta \leq 2 s \kappa + \delta.$
  Since $2\delta < \kappa$, the previous expression implies that $s' < s$.
  Similarly, by \eqref{eqn:delta-pm}, $\SCone$, and $\FTtwo$ we obtain
    $-2s\kappa \geq -(2s'+1)\kappa-\delta$, 
  such that, $s < s' + 1$.
  However, this contradicts $s' < s$.
  Thus $\FT$ cannot be satisfied at time~$t$ if $\SC$ is satisfied at time~$t$, as we assumed.
\end{proof}

\noindent
Combining Theorems~\ref{thm:gcs} and~\ref{lem:uncertainty},
  we finally obtain that an implementation of $\OGCS$ fulfilling the
  conditions in Theorems~\ref{lem:uncertainty}, maintains the skew bounds
  in Theorem~\ref{thm:gcs}.

\section{Decomposition into Modules}\label{sec:hardware}

To implement the $\OGCS$ algorithm in hardware, 
  we break down the distributed algorithm in this section into circuit modules.
Here we are concerned with the clock generation network that comprises an arbitrary number of nodes connected by links.
To focus on the clock generation circuitry, we do not
  discuss the circuitry for the data communication infrastructure that can be implemented using gradient clocking and communication links between clocked modules. 

We distinguish between the implementation of a node
  and the implementation of a link.~\footnote{For an in-depth survey of related work and state-of-the-art link-level FIFO buffer controllers, we refer the reader to~\cite{konstantinou2019mesochronous,bund2020synchronizer}.}
Per node, we have a \emph{tunable oscillator} that is responsible for maintaining
  the logical clock of a node, and a \emph{control module} that sets the local clock speed
  if the fast trigger \FT{} is fulfilled.
Per link, we have two \emph{phase offset measurement modules}, one for each node connected by the link,
  that measures the clock offset $\offset_w(t)$ of a node to its neighbor $w$ (see Figure~\ref{fig:fullimpl} for a high-level architecture comprising of the control module, the VCO as the tunable oscillator, and the measurement modules).

\myparagraph{Metastability-Containing Implementation}
The three modules form a control loop:
  Skews are measured and fed into the control module, which acts upon the tunable oscillator.
Any measurement circuit that round-wise measures a continuous variable, in our case, the skew, and outputs a digital representation can be shown to
  become metastable \cite{marino:general}.
In \cite{friedrichs:containing}, a technique to compute with such metastable or unstable signals
  was presented.
The term \emph{metastability-containing} circuit was coined for a circuit that
  guarantees that the provably minimal amount of metastability carries over from the inputs to the
  outputs.
In this work, we will build on this technique and design the circuitry to be
  metastability-containing:
(i) The measurement circuit outputs the minimal number of metastable outputs, (ii) the controller only produces metastable outputs if its inputs are unstable,
and (iii) the tunable oscillator frequency is bounded even in the presence of metastable/unstable signals.

\subsection{Tunable Oscillator}
The logical clock signal of node $v$ is derived from a tunable oscillator.
Each node is associated with its own oscillator that can be tuned in its frequency.
The tunable oscillator module has one input and one output port.
The binary input $\md_v$ controls the mode (slow/fast) of the oscillator.
The binary output $\clk_v$ is the oscillator's binary clock signal whose phase
  is the logical clock~$L_v$.
The oscillator has a maximum response time $\Tosc \geq 0$, by which it
  is guaranteed to change frequency according to the mode signal.
Formally, we require the following conditions to hold:
\begin{equation*}
  \calL(0)<\kappa\,.\tag{C1}\label{eqn:c1}
\end{equation*}
If mode signal of node $v$ is constantly $0$ for time $\Tosc$, the oscillator with
  output $\clk_v$ is in \emph{slow mode} at time~$t$:
\begin{align*}
    \forall t'\in[t-\Tosc,t].\, \md_v(t') = 0 \Rightarrow L_v(t)\in [1,1+\rho]\,.
  \tag{C2}\label{eqn:c2}
\end{align*}
If a mode signal is constantly $1$ for time $\Tosc$, the respective oscillator is in \emph{fast mode} at time~$t$:
\begin{align*}
    &\forall t'\in[t-\Tosc,t].\, \md_v(t') = 1 \Rightarrow\\
    &L_v(t)\in [1+\mu,(1+\mu)(1+\rho)]\,. \tag{C3}\label{eqn:c3}
\end{align*}
Otherwise, the respective oscillator is \emph{unlocked} at time~$t$:
\begin{align*}
    &\exists t',t''\in[t-\Tosc,t].\, \md_v(t') \neq \md_v(t'') \Rightarrow\\
    &L_v(t)\in [1,(1+\mu)(1+\rho)]\,.\tag{C4}\label{eqn:c4}
\end{align*}

\noindent
The requirements on the oscillator are as follows:
if the control signal is stable for
  $\Tosc$ time, the oscillator needs to guarantee the respective frequency.
  At any other time, it is not locked to a fixed mode and may run at any frequency
  between the slowest and fastest possible.
In particular, the unlocked mode may be entered when the mode signal is metastable, unstable,
  or transitioned recently, i.e., an oscillator that satisfies \eqref{eqn:c4} can cope with
  meta-/unstable inputs in the sense that it produces stable outputs.
We stress that a tunable oscillator satisfying \eqref{eqn:c4} is not pausible.

It is an essential requirement of the algorithm that the skew between two nodes cannot increase
  if the algorithm tries to reduce that skew.
We maintain the requirement $\mu > 2\rho$, ensuring that the phase offset
  between the two clocks cannot increase further when a clock in fast mode is
  chasing a clock in slow mode.

\subsection{Phase Offset Measurement Module}\label{sec:measmod}
To check whether the $\FT$ conditions are met, a node~$v$ needs to
  measure the current phase offset $L_w(t)-L_v(t)$ to each neighbor $w$.
This is achieved by a time offset measurement module between $v$ and each neighbor $w$.
Node $v$ has no direct access to $L_w(t)$ as propagation delays are prone to uncertainty.
Hence, a node can only estimate the offset to $w$, where the offset estimate is
  denoted by $\offset_w(t)$.

Inputs to the offset measurement are signals $\clk_v$ and $\clk_w$.
The outputs are denoted by $q^{\pm i}_w(t)$ for $i\in\{1,\ldots,\ell\}$.
They represent a unary encoding of $\offset_w(t)$ of length $2\ell$.
As mentioned before, the offset measurement module may produce metastable estimates.
We next discuss the module's specifications.

\myparagraph{Thresholds}
The algorithm does not require full access to the function
  $\offset_w(t)$, but only to whether $\offset_w(t)$ has reached
  one of the thresholds defined by \eqref{eqn:restateft1} and \eqref{eqn:restateft2}.
$\FT$ defines infinitely many thresholds, \ie, the algorithm has to check for each $s\in\IN$
  whether \eqref{eqn:restateft1} or \eqref{eqn:restateft2} is satisfied.

However, practically the system can only measure finitely many thresholds.
Since the algorithm guarantees a maximum local skew,
  there is a maximum $s$ until which the algorithm needs to check.
Let $\ell\in\IN$ be the largest number such that $(2\ell + 1)\kappa + \delta < \calL$,
  where $\calL$ is the upper bound on the local skew.
Then $\offset_w(t)$ is defined as a binary word of length $2\ell$.
The bits are denoted (from left to right) by $Q^{\ell}_w,\ldots,Q^{1}_w,Q^{-1}_w,\ldots,Q^{-\ell}_w$.
For $i\in\{1,\ldots,\ell\}$ each output bit $Q^{\pm i}_w(t)$ denotes whether
  $\offset_w(t)$ has reached the corresponding threshold.
For example, a module with $\ell=2$ has $4$ outputs $Q^{2}_w$, $Q^{1}_w$, $Q^{-1}_w$, and $Q^{-2}_w$
  corresponding to thresholds $-3\kappa-\delta$, $-\kappa-\delta$, $\kappa-\delta$, and $3\kappa-\delta$.
Each signal $Q^{\ell}_w,\ldots,Q^{1}_w,Q^{-1}_w,\ldots,Q^{-\ell}_w$ is a
  function of time.
For better readability, we omit the function parameter $t$ when it is clear from
  context.

\myparagraph{Decision Separator}
Any realistic hardware implementation of the offset measurement will have to
  account for setup/hold times of the registers it uses.
We dedicate the decision separator $\e$ to account for (small) additional setup/hold times,
  and the effect of a potentially metastable output in case a setup/hold time is violated.
A visualization of the decision seperator is given in \cref{fig:GCSmetas}.

\begin{figure}
  \centering
  \includegraphics[width=0.7\linewidth]{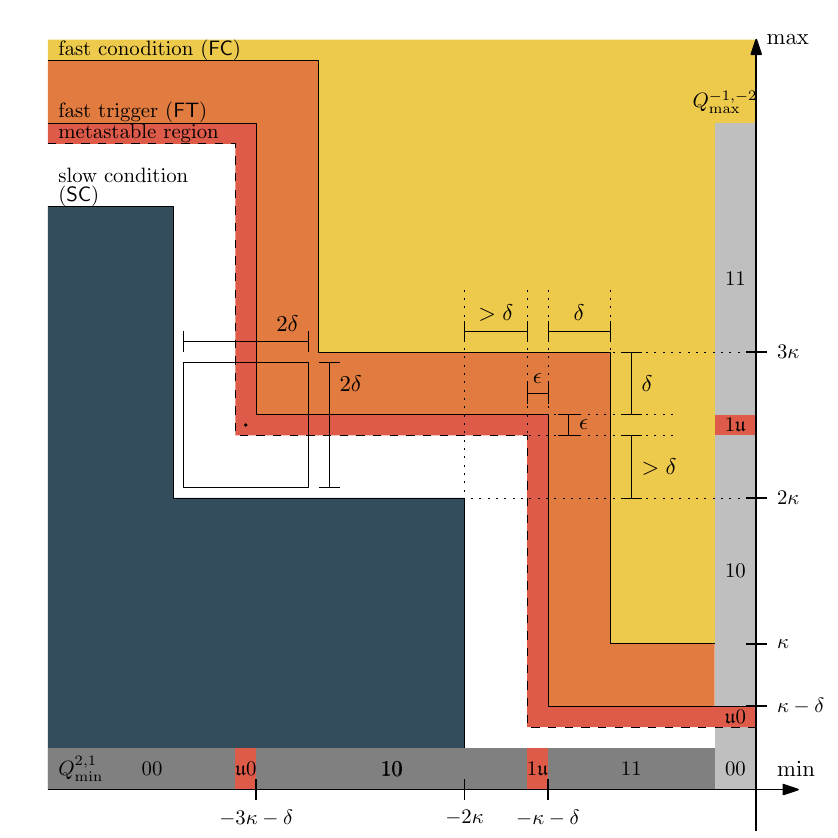}
  \caption{Update of \cref{fig:gcs_large} including the decision separator. As in \cref{fig:gcs_large} we visualize $\FC$ and $\SC$ with respect to $O_{\min}$ and $O_{\max}$.
  $\FT$ and the metastable region are visualized with respect to $\widehat{O}_{\min}$ and $\widehat{O}_{\max}$.
  Along the axis we show the encodings $Q_{\min}^{2,1}$ and $Q_{\max}^{-1,-2}$, marking also
  the effect of meta-/instability.}
  \label{fig:GCSmetas}
\end{figure}

We require that signal $Q^{\pm i}_w(t)$ is $1$ at time $t$ if the offset exceeds the $i$th threshold and
  we require that signal $Q^{\pm i}_w(t)$ is $0$ at time $t$ if the offset does not exceed the $i$th threshold.
When the offset is close to the threshold (within $\e$), then we allow that $Q^{\pm i}_w(t)$
  is unconstrained, \ie, $Q^{\pm i}_w(t)\in\{0,\mfu,1\}$,
  where $\mfu$ denotes a metastable or unstable value, e.g., a transition, a glitch,
  or a value between logical 0 and 1.
Formally, we define the module's outputs to fulfill the following:

\begin{definition}[decision separator]\label{def:decsep}
  Let $\e$ be a (small) timespan with $\kappa \gg \e  > 0$.
  At time $t$, we require the following constraints for all $i\in\{1,\ldots,\ell\}$.
  Signal $Q^{\pm i}_w(t)$ is set to $1$ if the offset estimate is larger than $\mp(2i-1) \kappa - \delta$.
  \begin{equation*}
    \begin{split}
      \offset_w(t) \geq -(2i-1) \kappa - \delta \ &\Rightarrow\ Q^i_w(t) = 1 \\
      \offset_w(t) \geq (2i-1) \kappa - \delta \ &\Rightarrow\  Q^{-i}_w(t) = 1
    \end{split}\tag{M1}\label{eqn:m1}
  \end{equation*}
  Signal $Q^{\pm i}_w(t)$ is set to $0$ if the offset measurement is smaller than $\mp(2i-1) \kappa - \delta - \e$.
  \begin{equation*}
    \begin{split}
      \offset_w(t) \leq -(2i-1) \kappa - \delta - \e \ &\Rightarrow\  Q^i_w(t) = 0 \\
      \offset_w(t) \leq (2i-1) \kappa - \delta - \e \ &\Rightarrow\  Q^{-i}_w(t) = 0
    \end{split}\tag{M2}\label{eqn:m2}
  \end{equation*}
  Otherwise, $Q^{\pm i}_w(t)$ is unconstrained, i.e., within $\{0,\mfu,1\}$.
\end{definition}

\noindent
Figure~\ref{fig:clk_meas} (middle) shows the timing of signals $Q^{-1}_w(t)$, $Q^{1}_w(t)$,
  and $Q^{2}_w(t)$ in relation to the clock of neighbor $w$.
When $\clkv$ transitions to $1$, the measurement module takes a snapshot
  of the outputs $Q^{\pm i}_w$.
In Figure~\ref{fig:clk_meas} (right), we show two examples.

Figure~\ref{fig:clk_meas} (left) depicts transitions of the signals $Q^{\pm i}_w(t)$.
The figure shows increasing $\widehat{O}_w$ (along the $x$-axis), resulting in
  more and more bits $Q^{\pm i}_w(t)$ flip to $1$.
The decision separator $\e$ is small enough that
  no two bits can flip at the same time.
If $\widehat{O}_w(t)=0$ we obtain $Q^{i}_w(t)=1$ and $Q^{-i}_w(t)=0$ 
for all $i\in\{0,\ldots,\ell\}$.

Figure~\ref{fig:clk_meas} (middle) also depicts transitions of the signals $Q^{\pm i}_w(t)$,
  but along the $x$-axis $L_v(t)$ increases while $L_w$ is fixed.
We mark time $\mathcal{L}_w$ at which $L_v(t)=L_w$.
A digital implementation is only able to measure the offset on a clock event, \eg,
  a rising clock transition.
Hence, $\mathcal{L}_w$ will be the time where $\clkw$ rises.
When $L_v(t)=L_w$, we have that $\widehat{O}_w(t)=0$, such that
  all bits $Q^{i}_w(t)=1$ and $Q^{-i}_w(t)=0$.
As $L_v(t)$ increases, $\widehat{O}_w(t)$ decreases.
Hence, in Figure~\ref{fig:clk_meas}, (middle) is a mirror image of (left).

\begin{figure*}
  \centering
  \begin{subfigure}{.33\linewidth}
    \includegraphics[width=\linewidth]{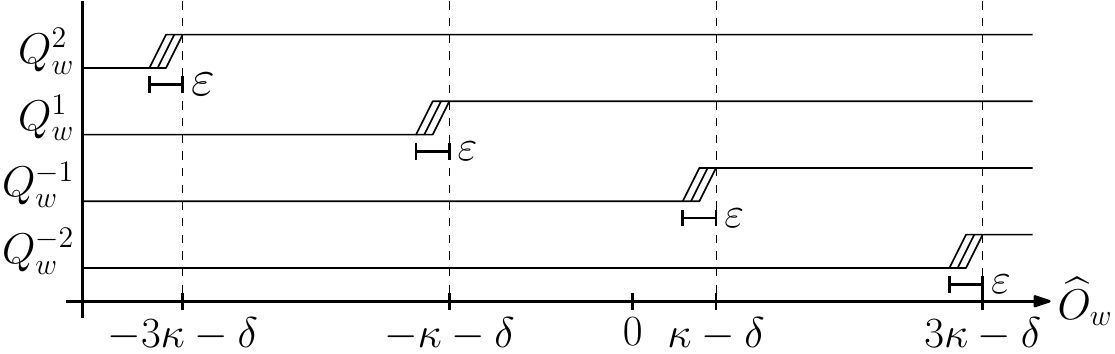}
  \end{subfigure}
  \begin{subfigure}{.33\linewidth}
    \includegraphics[width=\linewidth]{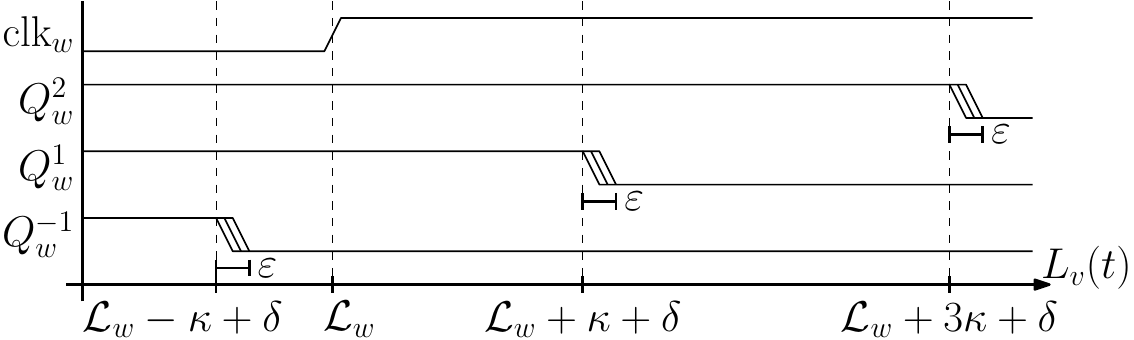}
  \end{subfigure}
  \begin{subfigure}{.33\linewidth}
    \includegraphics[width=\linewidth]{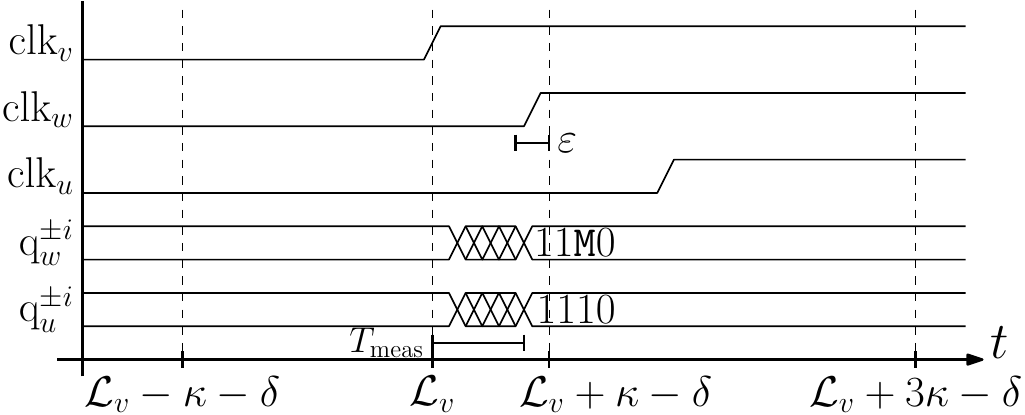}
  \end{subfigure}
  \caption{Timing diagrams of the measurement module including the decision separator.
  (left) Output bits $Q^{\pm i}_w$ relative to the actual measurement $\widehat{O}_w(t)$.
  (middle) Output bits $Q^{\pm i}_w(t)$ relative to the logical clock $L_v(t)$, assuming that $L_w(t)$ is known to $v$.
  (right) Example measurements from node $v$ to nodes $u$ and $w$, relative to Newtonian time $t$.}
  \label{fig:clk_meas}
\end{figure*}

\begin{example}
  Regarding Figure~\ref{fig:clk_meas} (right), a measurement module with $\ell = 2$ can have output $Q_u(t)=1100$ if
  $\kappa-\delta-\e\geq\offset_u(t)\geq-\kappa-\delta$. The output may become $Q_w(t)=11\mfu0$ if
  $\kappa-\delta>\offset_w(t)>\kappa-\delta-\e$.
\end{example}

In general, closely synchronized clocks have output $Q^{\pm i}_w(t)=1^\ell 0^\ell$.
If the clock of $v$ is ahead of $w$'s clock, the measurement $Q^{\pm i}_w(t)$ contains more $0$s than $1$s.
Similarly, if $v$'s clock is behind the clock of $w$, the outputs contain more $1$s than $0$s.
Further, at most one output bit is $\mfu$ at a time if the $\e$-regions in
  Figures~\ref{fig:clk_meas} (left) and (middle) do not overlap:

\begin{lemma}\label{lem:singleM}
  At every time $t$ there is at most a single $i\in\{1,\ldots\ell\}$ such that
  $Q^{\pm i}_w$ is unconstrained.
\end{lemma}
\begin{proof}
  Assume for $i>0$, that $Q^i_w(t)$ is unconstrained. Then we have that
  \begin{equation*}
    - (2i - 1)\kappa-\delta-\e < \widehat{O}_w(t) < -(2i - 1)\kappa - \delta\,.
  \end{equation*}
  Hence, for all $i'<i$ it holds that $\widehat{O}_w(t) < -(2i'-1)\kappa-\delta$,
  such that, by \cref{def:decsep}, $Q^{i'}_w=0$ and $Q^{-j}_w=0$ for all $j$.
  For all $i'>i$ we obtain $\widehat{O}_w(t) > -(2i'-1)\kappa-\delta-\e$,
  as $\kappa > \e$. Thus, by \cref{def:decsep}, $Q^{i'}_w=1$.
  An analogous argument shows that there is only one unconstrained bit if $Q^{-i}_w(t)$ is unconstrained.
\end{proof}

\myparagraph{Latency}
Besides setup/hold times, we have to account further for
  propagation delays.
Let $\Tmeas$ denote the maximum end-to-end latency of the
  measurement module, \ie, an upper bound on the elapsed time from when $Q^{\pm i}_w(t)$ is set,
  to when the measurements are available at the output.
More precisely, we require that if $Q^{\pm i}_w(t')$ is set to $x \in \{0,1\}$ for all $t'$ in $[t - \Tmeas, t]$,
  then the corresponding output $\text{q}^{\pm i}_w(t)$ is $x$ as depicted in Figure~\ref{fig:clk_meas} (right).

\subsection{Control Module}\label{sec:conts}
Each node $v$ is equipped with a control module.
Its input is the (unary encoded) time measurement, \ie, bits $\text{q}^{\pm i}_w(t)$,
  for each of $v$'s neighbors.
Output is the mode signal $\modes_v(t)$.

The control module is required to set the mode signal according to
  Algorithm $\OGCS$, i.e., to fast mode if $\FT$ is satisfied, otherwise
  the algorithm defaults to slow mode.
Denote by $\Tctr$ the maximum end-to-end delay of the control module circuit, \ie,
 the delay between its inputs (the measurement offset outputs) and its output
 $\modes_v(t)$.
We then require the following:
If $\OGCS$ continuously maps the algorithm's switch $\gamma(t)$ to $0$ for time $\Tctr$, then
  the output of the control module is $0$ at time~$t$:
\begin{align*}
  \forall t'\in[t-\Tctr,t].\, \gamma_v(t')=0 \Rightarrow \md_v(t) = 0\,.
  \tag{L1}\label{eqn:l1}
\end{align*}
If $\OGCS$ continuously maps the switch $\gamma(t)$ to $1$ for time $\Tctr$, then
  the output of the control module is $1$ at time~$t$:
\begin{align*}
  \forall t'\in[t-\Tctr,t].\, \gamma_v(t')=1 \Rightarrow
  \md_v(t) = 1\,.\tag{L2}\label{eqn:l2}
\end{align*}
Otherwise, the output is unconstrained, i.e., within $\{0,\mfu,1\}$.

Intuitively, $\FT$ triggers when there is an offset that crosses threshold $i$ and
  no other offset is below threshold $-i$ for some $i \in \{1,\hdots,\ell\}$.
Hence, we select the maximum and minimum of the offsets $Q^{\pm i}_w$ to all neighbors $w$.

Since the network also includes self-loops (\cf~\cref{sec:model}),
  each node, conceptually, measures the offset to itself.
The offset to self is always $0$.
In practice, that means that the maximum only needs to consider neighbors that are ahead
  and the minimum only needs to consider neighbors that are behind.
For $i\in\{1,\hdots,\ell\}$, signals $Q^{-i}_w(t)$ indicate whether node $w$ is ahead
  and similar bits $Q^{i}_w(t)$ indicate whether $w$ is behind.
Thus, the $\ell$-bit encodings of maximum ($Q^{-i}_{max}(t)$) and minimum ($Q^{i}_{min}(t)$) are computed as
\begin{align*}
  Q^{-i}_{max}(t)&\coloneqq\bigvee\{ Q^{-i}_w(t) \mid w \text{ is neighbor of } v \}\,,\\
  Q^{i}_{min}(t)&\coloneqq\bigwedge\{ Q^{i}_w(t) \mid w \text{ is neighbor of } v \}\,.
\end{align*}
As $\FT$ is satisfied if $Q^{-i}_{max}(t)$ and $Q^{i}_{min}(t)$ are both $1$ for any $i$ in $\{1,\hdots,\ell\}$.
Signal $\modes_v(t)$ is computed by
\begin{equation*}
  \modes_v(t)\coloneqq\bigvee\{Q^{-i}_{max} \land Q^{i}_{min} \mid i \in \{ 1,\hdots,\ell\} \}\,.
\end{equation*}

\myparagraph{Metastability-containment}
Any metastability-containing implementation of $\modes_v(t)$ has the following properties:
  (i) If the slow condition is satisfied, then $\modes_v(t)=0$
  (ii) if the fast condition is satisfied, then $\modes_v(t)=1$
  (iii) if no condition is satisfied then $\modes_v(t)$ may output $\metas$.
For a formal definition of metastability-containment, we refer the reader to \cite{friedrichs:containing}.
In \cref{sec:impl}, we present a metastability-containing implementation of the
  control module.

\subsection{$\CGCS$ Algorithm}\label{sec:cgcs}

\myparagraph{Clocked Algorithm}\label{par:clockedalg}
We are now in the position to assemble the modules into the so-called Clocked Gradient
  Clock Synchronization ($\CGCS$) algorithm (see \cref{alg:dgcs}).
In the following we prove Theorem~\ref{lem:implements}, showing that the $\CGCS$ algorithm
  implements the $\OGCS$ algorithm, and hence maintains tight skew bounds.
For the measurement module, we defined a possibly metastable assignment if the signal changes within an $\e$ window
  during which it is assigned.
We denote the assignment with propagation delay and possibly
  metastable result by~$\leftarrow_{\mfu}$.

\begin{algorithm}
  \begin{algorithmic}[1]
    \On{each clock tick, at time $t_{clk}$}
      \For{each $i \in \{ 1,\hdots,\ell \}$}
        \For{each adjacent node $w$}
          \State $q^{i}_w \leftarrow_{\mfu} \offset_w(t_{clk}) \geq -(2i-1)\kappa - \delta$
          \State $q^{-i}_w \leftarrow_{\mfu} \offset_w(t_{clk}) \geq (2i-1)\kappa - \delta$
        \EndFor
      \EndFor
    \EndOn
    \On{each time $t$}
      \State $q^{i}_{\min} \leftarrow_{\mfu} \bigwedge\{ q^{i}_w(t) \mid w \text{ is neighbor of } v \}$
      \State $q^{-i}_{\max} \leftarrow_{\mfu} \bigvee\{ q^{-i}_w(t) \mid w \text{ is neighbor of } v \}$
      \State $\modes_v \leftarrow_{\mfu} \bigvee\{q^{i}_{\min}(t) \land q^{-i}_{\max}(t) \mid i \in \{ 1,\hdots,\ell\} \}$
    \EndOn
  \end{algorithmic}
  \caption{Clocked algorithm $\CGCS$ at $v$.
  The assignment $\leftarrow_{\mfu}$ denotes a possibly unstable assignment.}
  \label{alg:dgcs}
\end{algorithm}

\myparagraph{The $\CGCS$ implements the $\OGCS$ Algorithm}
An essential difference of the $\CGCS$ to the continuous time
  $\OGCS$ algorithm is that measurements are performed only at discrete
  clock ticks.
We will, however, show that the clocked algorithm implements the $\OGCS$ algorithm,
  with a properly chosen measurement error $\delta$ that accounts for the fact that
  we measure clock skew only at discrete points in time rather than continuously.

For that purpose, we denote the maximum end-to-end latency of the computation by $\Tmax$. 
This end-to-end latency combines the delays of the three modules, \ie, $\Tmax = \Tmeas+\Tctr+\Tosc$.
Thus, $\Tmax$ is the time it takes from a rising clock edge until the oscillator guarantees a stable rate.
For a simple implementation, $\Tmax$ naturally becomes a lower bound on the clock period.
Designs with a clock period beyond $\Tmax$ are possible when buffering measurements
  and mode signals.

\begin{example}
A timing diagram with the module outputs and the clock rate is given in \cref{fig:tmax}.
The offset measurement switches from $1100$ (close to synchronous)
  to $1110$ ($v$ lagging behind $w$) and causes the oscillator to go to fast mode.
\end{example}

\begin{figure}
  \centering
  \includegraphics[width=.7\linewidth]{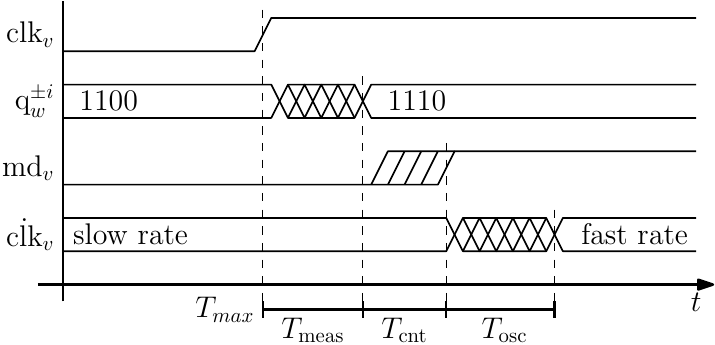}
  \caption{Example timing diagram of the control module, depicting signals $\clkv$, $q^{\pm i}_w$, $\md_v$ and rate $\dot{\clk}_v$ over Newtonian time $t$.}
  \label{fig:tmax}
\end{figure}

\noindent
We are now in the position to relate the module delays to~$\delta$.
We split $\delta$ into two parts, the \emph{propagation delay uncertainty} and the
  \emph{maximum end-to-end latency}.
The propagation delay uncertainty accounts for variations in the time a signal takes to propagate from a node's oscillator to the measurement module of its neighbors.
Suppose clock signals arrive at the measurement module with a larger or smaller
delay than expected (usually due to variation in the fabrication process or
environmental influences), then the module may measure larger or smaller offsets.
We denote the propagation delay uncertainty by $\delta_0$.

The second source of error is the drift of the clocks when not measuring.
The offset is measured once per clock cycle and it is used until the next measurement
is made. During this time, the actual offset may change due to different modes
and drift of oscillators. We denote the duration of a clock cycle (in slow mode
with no drift) by $\Tclk$. The maximum difference in rate between
any two logical clocks is bounded by $(1+\rho)(1+\mu)-1=\rho + \mu + \rho \mu$.
Thus, the maximum change of the offset during a clock cycle is at most
$$(\rho + \mu + \rho \mu)(\Tclk+\Tmax)\:.$$ This is the second contribution to
the uncertainty of the measurement.
Summing up both contributions, the measurement error becomes
$$\delta = \delta_0 + (\rho + \mu + \rho \mu) \cdot (\Tclk+\Tmax)\:.$$
Formally, we have the following:

\begin{lemma}\label{lem:clkdelta}
  Let $\delta = \delta_0 + (\rho + \mu + \rho \mu) \cdot \Tclk$, then $\CGCS$
  satisfies Inequality~\eqref{eqn:delta-pm} at all times $t$.
\end{lemma}
\begin{proof}
  The algorithm measures the offset $\offset_w(t)$ at each clock tick. Hence,
  we show that between two clock ticks the uncertainty never grows beyond $\delta$.
  Let $t_{clk}$ and $t'_{clk}$ be two consecutive clock ticks at node $v$.
  By the specification above, the measurement at time $t_{clk}$ has precision $\delta_0$, such that
  \begin{equation*}
    \abs{\offset_w(t_{clk})-(L_w(t_{clk})-L_v(t_{clk}))}\leq\delta_0\,.
  \end{equation*}
  During time interval $[t_{clk},t'_{clk})\leq\Tctr$ the clock rates may be different
  for neighbors. The difference between logical clocks grows at most by
  $(1+\rho)(1+\mu)\cdot\Tclk-\Tclk=(\rho + \mu + \rho \mu) \cdot \Tclk$, such that
  for $t\in[t_{clk},t'_{clk})$,
  \begin{align*}
    & \abs{\offset_w(t) - (L_w(t) - L_v(t))} \\
    & \leq \abs{\offset_w(t_{clk})-(L_w(t_{clk}) - L_v(t_{clk}))}+(\rho + \mu + \rho \mu) \cdot \Tclk\\
    & \leq \delta_0 + (\rho + \mu + \rho \mu)\cdot\Tclk\,.
  \end{align*}
  Hence, at every time the error is at most $\delta$, such that Inequality~\eqref{eqn:delta-pm} is satisfied.
\end{proof}

\noindent
We are now in the position to prove the section's main result:
  under certain conditions on the algorithm's parameters,
  $\CGCS$ implements $\OGCS$.
If, in addition, the algorithm's parameters fulfill the conditions
  in Theorem~\ref{lem:uncertainty}, it follows that
  the skew bounds from the \textsf{GCS} algorithm apply to $\CGCS$.

\begin{theorem}
  \label{lem:implements}
  Algorithm $\CGCS$ is correct, i.e., it maintains the skew bounds in Theorem~\ref{thm:gcs},
    if its parameters $\e$, $\delta_0$, and $\rho$ fulfill
    constraints \eqref{eqn:c1}\,--\,\eqref{eqn:c4}, \eqref{eqn:m1}, \eqref{eqn:m2},
    \eqref{eqn:l1}, and \eqref{eqn:l2} and parameters $\mu$, $\kappa$, and $\ell$
    are chosen according to Theorem~\ref{lem:uncertainty}.
\end{theorem}
\begin{proof}
By choosing $\delta$ as in Lemma~\ref{lem:clkdelta}, 
  \Cref{eqn:delta-pm} is satisfied.
A bounded local skew at all times $t \geq 0$ follows from \eqref{eqn:c1}
  and Theorem~\ref{lem:uncertainty}.
This further implies the finiteness of parameter $\ell$.
For the correct choice of $\ell$, lines $2$--$5$ in $\CGCS$ correspond to
  lines $2$ and $3$ of $\OGCS$ according to \eqref{eqn:m1} and
  \eqref{eqn:m2}.
Given a metastability-containing implementation, line $7$ (respectively\ $8$)
  of $\CGCS$ corresponds to line $4$ (respectively\ $5$) of $\OGCS$ according
  to Boolean logic.
Line $9$ of $\CGCS$ corresponds to lines $6$--$12$ of $\OGCS$, where
  switching to fast (respectively\ slow) mode is ensured by \eqref{eqn:l1} and
  \eqref{eqn:c2} (respectively\ \eqref{eqn:l2} and \eqref{eqn:c3}).
In case of a metastable assignment, \eqref{eqn:c4} ensures a correct behavior of the oscillator.

If constraints \eqref{eqn:c1}\,--\,\eqref{eqn:c4}, \eqref{eqn:m1},
  \eqref{eqn:m2}, \eqref{eqn:l1}, and \eqref{eqn:l2} are fulfilled,
  then $\CGCS$ implements $\OGCS$.
By Theorem~\ref{lem:uncertainty}, every algorithm that implements
  $\OGCS$ and satisfies the constraints in Theorem~\ref{lem:uncertainty}
  maintains the skew bounds of Theorem~\ref{thm:gcs}.
\end{proof}

\section{Hardware Implementation}\label{sec:impl}

We next present a hardware implementation of the $\CGCS$ algorithm, which we refer to as \BFM. 
We then discuss its performance and how the system's parameters affect the achieved skews.
For the latter, we designed an ASIC in the \SI{15}{\nano\metre} FinFET-based NanGate
  \textsc{ocl}~\cite{martins2015open} technology.
The design is laid out and routed with Cadence Encounter,
  which is also used for the extraction of parasitics and timing.
Local clocks run at a frequency of approximately \SI{2}{\giga\hertz},
  controllable within a factor of $1 + \mu \approx 1 + 10^{-4}$.
We use a larger factor $\mu$ to make the interplay of $\rho$ and $\mu$ better
  visible.
We compile two systems of $4$ \resp $7$ nodes connected in a line.
To resemble a realistically sparse spacing of clock regions,
  we placed nodes at distances of \SI{200}{\micro\metre}.
Hence, the PALS systems shown in the simulations are designed to 
cover floorplans of width $\SI{200}{\micro\metre}$ and length 
  $\SI{800}{\micro\metre}$ respectively\ $\SI{1.4}{\milli\metre}$.

\myparagraph{Offset Measurement}
\Cref{fig:measdline} shows a linear TDC-based circuitry for the
  module which measures the time offsets between nodes $v$ and $w$.
Buffers are used as delay elements for incoming clock pulses.
The offset is measured in steps of $2\kappa$, hence, buffers in the upper delay line
  have a delay of $2\kappa$.
The delay line is tapped after each buffer for corresponding $Q^{\pm i}_w$.
A chain of flip-flops takes a snapshot of the delay line by sampling the taps.
We require $Q^{-i}_w=0$ and $Q^{i}_w=1$, for all $i$, when
$\widehat{O}_w\geq -\kappa-\delta$ and $\widehat{O}_w\leq \kappa-\delta-\e$ according 
  to \eqref{eqn:m1} and \eqref{eqn:m2}.
Thus, we delay $\clk_v$ by $5\kappa+\delta+\e$.
The decision separator $\e$ accounts for the critical setup/hold window of the
  flip-flop.

\begin{example}
  If both clocks are perfectly synchronized, \ie, $L_v = L_w$, then the state of the flip-flops will
  be $Q^{3}_wQ^{2}_wQ^{1}_wQ^{-1}_wQ^{-2}_wQ^{-3}_w = 111000$ after a rising transition of $\clkv$.
  Now, assume that clock $w$ is ahead of clock $v$, say by a small $\e > 0$ more than
  $\kappa + \delta$, \ie, $L_w = L_v + \kappa - \delta + \e$.
  For the moment assuming that we do not make a measurement error,
  we get $\offset_w = L_w - L_v = \kappa - \delta + \e$.
  From the delays in \cref{fig:measdline} one verifies that in this case,
  the flip-flops are clocked before clock $w$ has reached the second flip-flop with
  output $Q^1_w$, resulting in a snapshot of $110000$.
  Likewise, an offset of $\offset_w = L_w - L_v = 3\kappa - \delta + \e$
  results in a snapshot of $100000$.
\end{example}

\begin{figure}
  \centering
  \includegraphics[width=.9\linewidth]{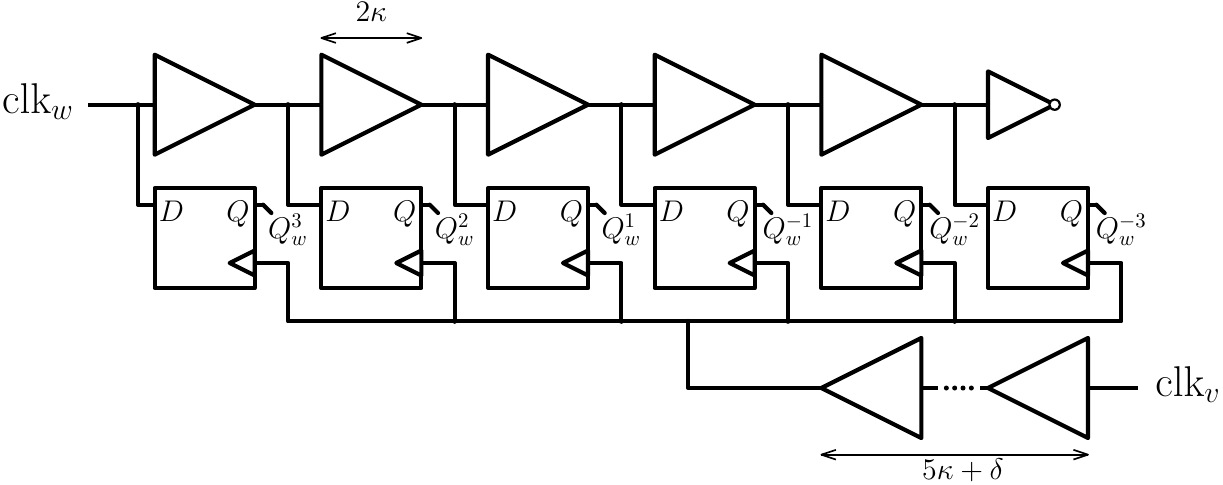}
  \caption{Schematic of the time offset measurement module for $\ell=3$.}
  \label{fig:measdline}
\end{figure}

\myparagraph{Control Module}
Given node $v$'s time offsets to its neighbors in unary encoding, the control module
  computes the minimal and maximal threshold levels which have been reached.
The circuit in Figure~\ref{fig:controlm} implements the control module for $3$ neighbors $w_1$, $w_2$, and $w_3$.
As described in \cref{sec:conts}, we only need to compute the maximal value of bits $Q^{-i}_w$
and the minimal value of bits $Q^{i}_w$; which can be easily computed
  by an $\OR$ respectively $\AND$ over all neighbors.

Given the maximal and minimal values, the circuit in
  \cref{fig:mode} computes $\FT$ and sets $\md_v$ to $1$ if it holds.

\begin{figure}
  \centering
  \begin{subfigure}{.33\linewidth}
    \centering
    \includegraphics[width=.8\linewidth]{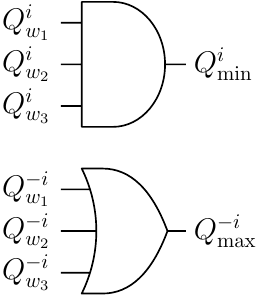}
    \caption{minimum and maximum offset, for $i\in\{1,2,3\}$}
    \label{fig:minmax}
  \end{subfigure}
  \begin{subfigure}{.6\linewidth}
    \centering
    \includegraphics[width=\linewidth]{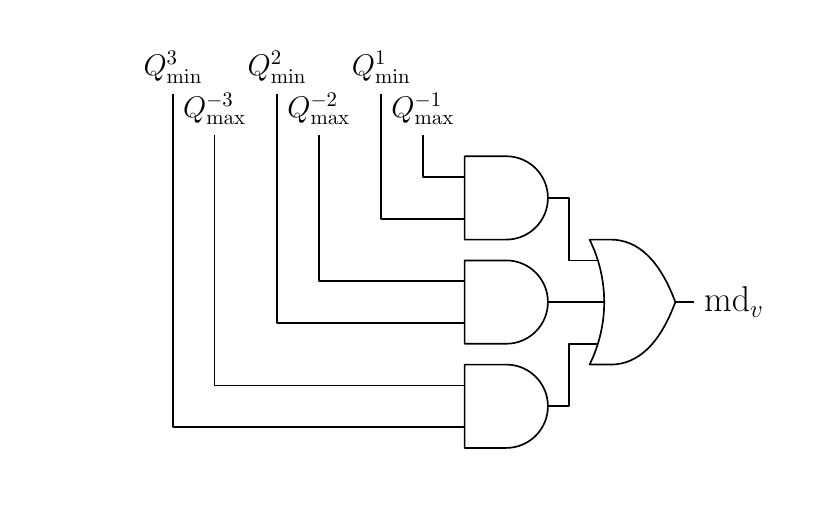}
    \caption{mode signal}
    \label{fig:mode}
  \end{subfigure}
  \caption{Schematics of the control module for three neighbors ($w_1$, $w_2$, and $w_3$).}
  \label{fig:controlm}
\end{figure}

\myparagraph{Metastability-Containing Control Module}
As described in \cref{sec:hardware}, the inputs to the control module may be metastable, i.e., unknown or even oscillating signals between ground and supply voltage. 
It remains for us to show that the control module, given in Figure~\ref{fig:controlm},
  fulfills specifications \eqref{eqn:l1} and \eqref{eqn:l2}.
In particular, we need to ensure that the specifications are met for
  metastable inputs.

The circuit in Figure~\ref{fig:controlm} follows from the definition of $\modes_v$
  in \cref{sec:conts}, when replacing each conjunction (respectively\ disjunction) by an $\AND$
  (respectively\ $\OR$) gate.
Due to the masking properties of $\AND$ and $\OR$ gates, the output $\modes_v$ can
  only become metastable when there is an $i$ such that one of $Q^{i}_{\min}$
  or $Q^{-i}_{\max}$ is $\metas$ and the other is $1$ or $\metas$.
This is only the case when \cref{eqn:l1} and \cref{eqn:l2} do not apply.
It follows that the output of the control module can only become metastable when
  the mode signal is unconstrained and the conditions are met.

\myparagraph{Tunable Oscillator}
As a local clock source, we use a ring oscillator inspired by
  the starved inverter ring presented in~\cite{ghai2009design}.
We use a ring of inverters, where some inverters
  being current-starved-inverters, to set the frequency to either fast mode or slow mode.
Nominal frequency is around \SI{2}{\giga\hertz}, controllable by a factor
  $1 + \mu \approx 1 + 10^{-4}$ via the $\modes_v$ signal.
For our simulations, we choose $\rho \approx \mu/10\approx 10^{-5}$, assuming a stable oscillator.
While this requirement poses a challenge to an oscillator design, it can be relaxed in different ways: (i) By choosing larger parameters $\mu$ and $\rho$ such that their ratio remains fixed, i.e., $\mu/\rho=10$, this forces us to choose a larger $\kappa$, and hence requires to measure larger time offsets.  
This is at the cost of a larger skew and circuit (this follows from combining Theorem~\ref{thm:gcs}, Theorem~\ref{lem:uncertainty}, and Lemma~\ref{lem:clkdelta}).
(ii) By locking the local oscillators to a central quartz oscillator. The problem is different from building a balanced clock tree since the quartz's skew can be neglected here.
(iii) We conjecture that with a refined analysis of the algorithm: rather than absolute drift, the drift with respect to a neighboring oscillator is determinant. Neighboring oscillators show reduced drift due to common cause effects.

The tunable ring oscillator comes with the advantage that for any input voltage 
  it runs at a speed between fast and slow mode, hence, \eqref{eqn:c4} is 
  satisfied.
In the following paragraph, we define an upper bound on $\Tosc$ such that
  constraints \eqref{eqn:c2} and \eqref{eqn:c3} are satisfied.

\myparagraph{Timing Parameters}
We next discuss how the modules' timing parameters relate to the extracted
  physical timing of the above design.

The time $\Tosc$ required for switching between oscillator modes is about the
  delay of the ring oscillator, which in our case is about
  $1/(2\cdot \SI{2}{\giga\hertz}) = \SI{250}{\pico\second}$.
An upper bound on the measurement latency ($\Tmeas$) plus the controller
  latency ($\Tctr$) is given by a clock cycle (\SI{500}{\pico\second}) plus
  the delay (\SI{25}{\pico\second})
  of the circuitry in Figure~\ref{fig:controlm}.
In our case, delay extraction of the circuit yields
  $\Tmeas+\Tctr < \SI{500}{\pico\second} + \SI{25}{\pico\second}$.
We thus have, $\Tmax < \Tmeas + \Tctr + \Tosc = \SI{775}{\pico\second}$.

The uncertainty, $\delta_0$, in measuring if $\offset_w$ has reached a certain threshold is given by the
  uncertainties in latency of the upper delay chain plus the lower delay chain in \cref{fig:measdline}.
For the described naive implementation using an uncalibrated delay line, this would be problematic:
Extracting delays from the design after layout, the constraints
  from \cref{lem:uncertainty} were met for delay uncertainties of $\pm 1\%$, but not for the
  $\pm 5\%$ we targeted.
We thus redesigned the Offset Measurement circuit as described in the following.

\myparagraph{Improved Offset Measurement}
\Cref{fig:bettermeas} shows an improved TDC-type offset measurement circuit that
  does not suffer from the problem above.
Conceptually, the TDC of the node $v$ that measures offsets \wrt node $w$ is
  integrated into the local ring oscillator of the neighboring node $w$.
If $w$ has several neighbors, \eg, up to $4$ in a grid,
  they share the taps but have their own flip-flops within the node $w$.

\Cref{fig:bettermeas} presents a design for $\ell=2$ with $4$ taps and a single neighbor $v$.
In our hardware implementation we set $\ell=2$, as even for $\mu/\rho=10$
  this is sufficient for networks of diameter up to around $80$
  (see how to choose this set of thresholds in the specification of this module in \cref{sec:hardware}).
The gray buffers at the offset measurement taps decouple the load of the remaining circuitry.
At the bottom of the ring oscillator, an odd number of starved inverters are used to set the slow or fast mode for node $w$.
The delay elements at the top are inverters instead of buffers to achieve a latency of $\kappa = \SI{10}{\pico\second}$.
We inverted the clock output to account for the negated signal at the tap of clock $w$ at the top.

\begin{figure}
\centering
  \includegraphics[width=0.6\linewidth]{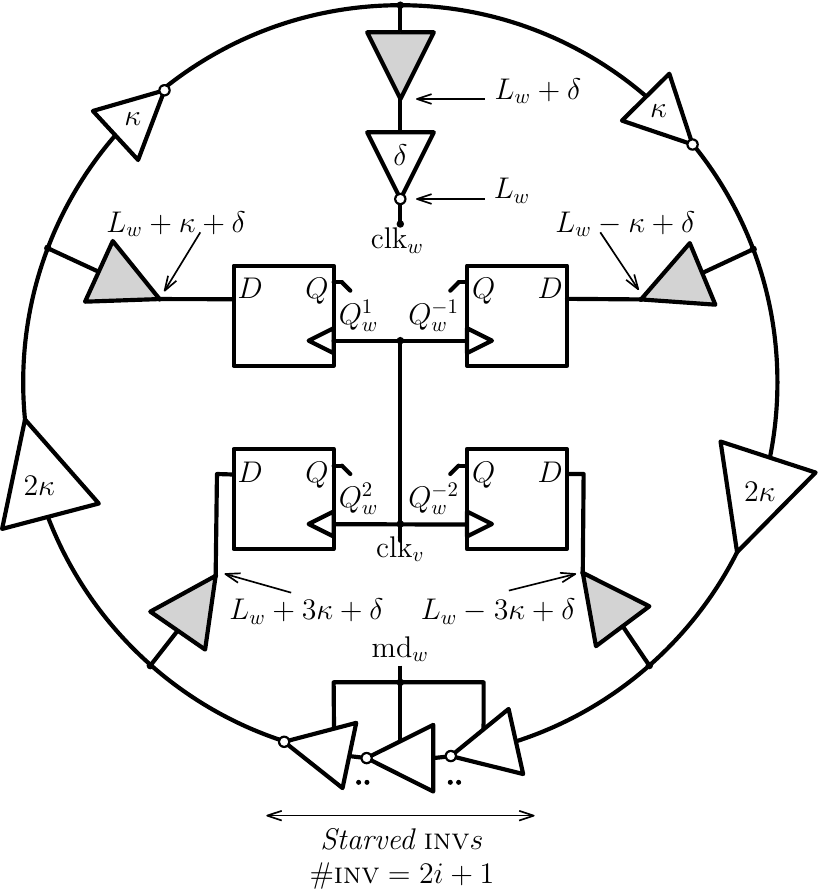}
  \caption{Schematic of the improved offset measurement implementation.
  Node $v$'s offset measurement is integrated into $w$'s ring oscillator.
  Labels of the delay elements denote their delay.
  We also annotate the measured phase offsets.}
  \label{fig:bettermeas}
\end{figure}

When integrating the measurement into the ring oscillator, the constraints
  \eqref{eqn:c2}\,--\,\eqref{eqn:c4} and \eqref{eqn:m1}, \eqref{eqn:m2} are still
  met for a suitable choice of $\Tosc$.
Integration of the TDC into $w$'s local ring oscillator greatly reduces uncertainties
  at both ends: (i) the uncertainty at the remote clock port (of node $w$) is removed to a large extent since the
  delay elements which are used for the offset measurements are part of $w$'s oscillator,
  and (ii) the uncertainty at the local clock port is greatly reduced by removing the delay line
  of length $5\kappa + \delta$.
The remaining timing uncertainties are the latency from taps to the D-ports of the flip-flops
  and from clock $v$ to the $\clkv$-ports of the flip-flop.
Timing extraction yielded $\delta_0 < \SI{4}{\pico\second}$
  in presence of $\pm 5\%$ gate delay variations.

\myparagraph{Full Hardware Implementation}
\begin{figure}
\centering
  \includegraphics[width=\linewidth]{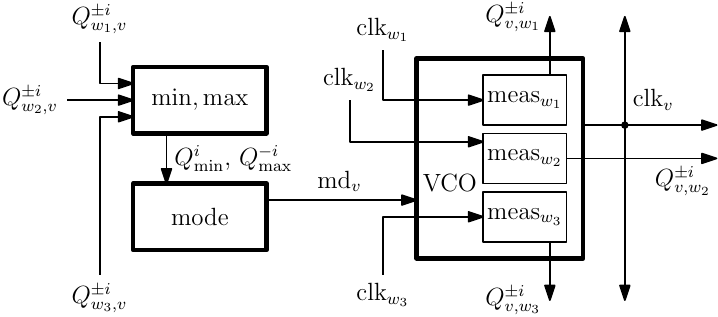}
  \caption{Schematic of a node $v$ with neighbors $w_1$, $w_2$, and $w_3$, using the improved offset measurement integrated in the $\operatorname{VCO}$. Note that measurement $Q^{\pm i}_{a,b}$ is the difference between nodes $a$ and $b$, as measured by $b$.}
  \label{fig:fullimpl}
\end{figure}
We depict in \cref{fig:fullimpl} the schematic of a node with three neighbors.
From Theorem~\ref{lem:uncertainty}, we obtain
  $\kappa \approx \SI{10}{\pico\second}$ and $\delta \approx \SI{5}{\pico\second}$
  which matched the previously chosen latencies of the delay
  elements.
Applying Theorem~\ref{thm:gcs}, finally, yields the global and local skew
  bounds:
  $\calG(t) \leq 1.223 \kappa D = 12.23 D \,\si{\pico\second}$  and
  $\calL(t) \leq (\ceil{ \log_{10} (1.223 D)} + 1)\kappa$.
For our design with diameter $D=3$ this makes
  a maximum global skew of \SI{36.69}{\pico\second}
  and a maximum local skew of $2\kappa = \SI{20}{\pico\second}$.
With regard to diameter $D=6$, we obtain
  a maximum global skew of \SI{73,38}{\pico\second}
  and a maximum local skew of $2\kappa = \SI{20}{\pico\second}$.

As described above, \BFM\ satisfies constraints \eqref{eqn:c2}
  \,--\,\eqref{eqn:c4}, \eqref{eqn:m1}, \eqref{eqn:m2}, \eqref{eqn:l1}, and 
  \eqref{eqn:l2}.
Parameters $\e$, $\delta_0$, and $\rho$ are restricted by the 
  technology.
For our choices of $\ell$, $\mu$, and $\kappa$ \BFM\ satisfies the conditions
  of Theorem~\ref{lem:uncertainty}.
Thus, from Theorem~\ref{lem:implements} it follows that our implementation
  maintains the skew bounds of the \textsf{GCS} algorithm.
\begin{corollary}\label{col:correctness}
  Given \eqref{eqn:c1}, \BFM\ is an implementation of $\CGCS$.
  Hence, it maintains global skew $\calG(t)\leq\calO(\delta D)$ and local skew
    $\calL(t)\leq\calO(\delta\log_{\mu/\rho} D)$.
\end{corollary}

\begin{remark}
\begin{inparaenum}[(i)]
    \item Note that this section gives a ``recipe'' for realizing a \BFM\ system for any network, e.g., how the hardware parameters affect the local and global skews. Also note that because a \BFM\ system implements a \CGCS\ algorithm, all skews are provably and deterministically guaranteed.
    
    \item Considerably larger systems, \eg, a grid with side length of $W = 32$ nodes and diameter $D=2W-2 = 62$, still are guaranteed to have a maximum local skew of $2\kappa = \SI{20}{\pico\second}$. If we choose $\mu=10^{-3}$, the base of the logarithm in the skew bound increases from $10$ to $100$.
\end{inparaenum}
\end{remark}

\section{Simulations}\label{sec:simu}

We ran \textsc{spice} simulations of the post-layout extracted design with Cadence Spectre.
Simulated systems had $4$ and $7$ nodes arranged in a line, as described in \cref{sec:impl}.
We chose a line setup since it allows us to compare local (between neighbors) versus
  global (typically the line ends) skew best.
Nodes are labeled $0$ to $3$ (respectively $6$).
For the simulations, we set $\mu = 10\rho$ (instead of $100\rho$), resulting in a slower
  decrease of skew, to observe better how the skew is removed.
Operational corners of the \textsc{spice} simulations were $\SI{0.8}{\volt}$ supply voltage and $\SI{27}{\degreeCelsius}$ temperature. 

Simulation with a small initial skew yields a peak power of $\SI{20.25}{\milli\watt}$ during stabilization and an average power of $\SI{5.26}{\milli\watt}$.
The performance measure of our system is given by the quality of the local skew.
We discuss in detail the local skew for different set-ups in the next section.

\subsection{\textsc{spice} Simulations on a $4$ Node Topology}\label{sec:sim4line}

\myparagraph{Scenarios}
We designed three simulation scenarios with different initial skews
  that demonstrate different properties of the algorithm:
\textsc{ahead}: node~$1$ is initialized with an offset of
  \SI{40}{\pico\second} ahead of all other nodes,
\textsc{behind}: node~$1$ is initialized with an offset of
  \SI{40}{\pico\second} behind all other nodes, and
\textsc{gradient}: nodes are initialized with small skews on each edge,
  that sum up to a large \SI{105}{\pico\second} global skew.
Simulation time for all scenarios is
  \SI{1000}{\nano\second} ($\approx$~$2000$ clock cycles).

\begin{figure*}
  \begin{subfigure}{0.245\linewidth}
    \centering
    \includegraphics[width=\linewidth,clip,trim=65 0 61 0]{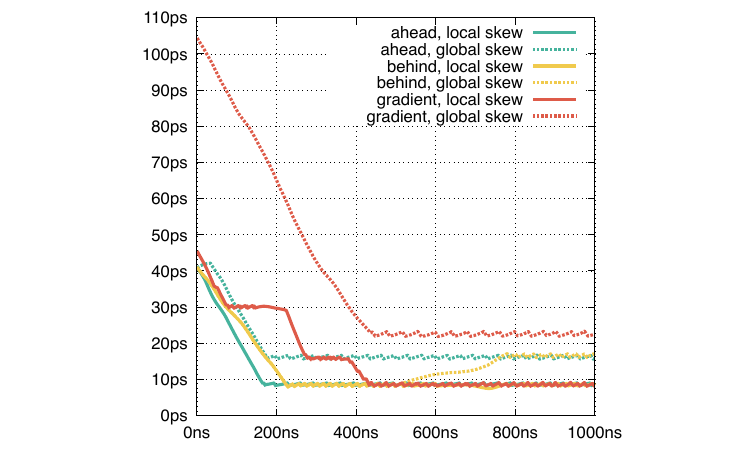}
    \caption{Maximum local skew and global skew for all scenarios}
    \label{fig:sim_skews_globloc}
  \end{subfigure}
  \begin{subfigure}{0.245\linewidth}
    \centering
    \includegraphics[width=\linewidth,clip,trim=65 0 61 0]{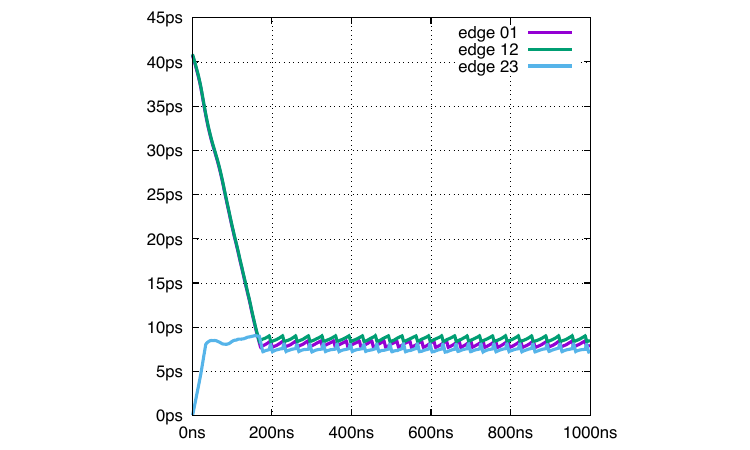}
    \caption{Scenario\\\textsc{ahead}}
    \label{fig:sim_skews_ahead}
  \end{subfigure}
  \begin{subfigure}{0.245\linewidth}
    \centering
    \includegraphics[width=\linewidth,clip,trim=65 0 61 0]{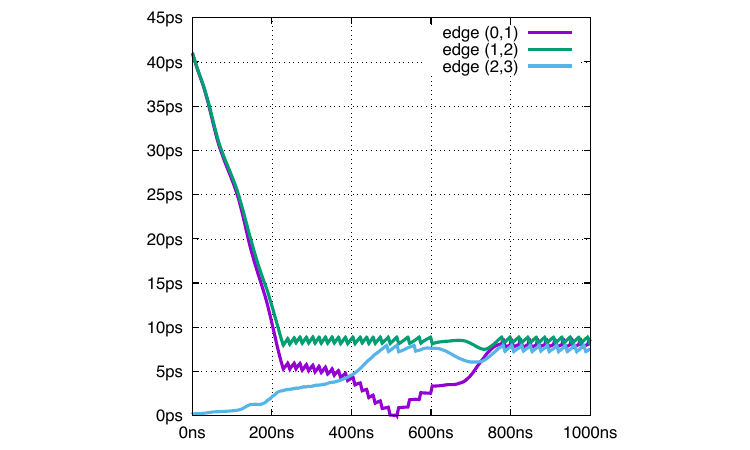}
    \caption{Scenario\\\textsc{behind}}
    \label{fig:sim_skews_behind}
  \end{subfigure}
  \begin{subfigure}{0.245\linewidth}
    \centering
    \includegraphics[width=\linewidth,clip,trim=65 0 61 0]{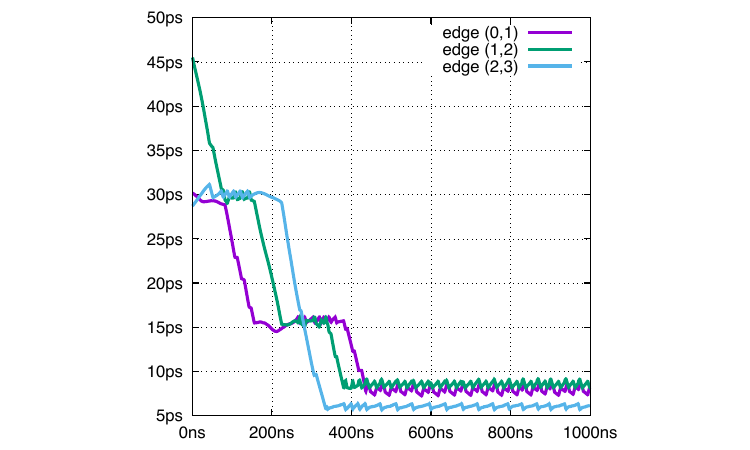}
    \caption{Scenario\\\textsc{gradient}}
    \label{fig:sim_skews_grad}
  \end{subfigure}
  \caption{Skews on edges $(0,1)$, $(1,2)$, and $(2,3)$ in the $4$ node topology, showing the skew (y-axis) over simulated time (x-axis) for SPICE simulations of scenarios \textsc{ahead}, \textsc{behind}, and \textsc{gradient}.}
  \label{fig:sim_skews}
\end{figure*}

\Cref{fig:sim_skews_globloc} depict the local and global
  skews of all scenarios.
Observe that all local skews decrease until they
  reach less than \SI{9}{\pico\second}.
The local skew then remains in a stable region.
This is well below our worst-case bound of \SI{20}{\pico\second} on the local skew.
We observe that the global skew slightly increases at the beginning of scenario \textsc{ahead}
  and after roughly \SI{500}{\nano\second} in scenario \textsc{behind}.

\myparagraph{One Node Ahead}
\Cref{fig:sim_clocks} shows the clock signals of nodes~$0$ to $3$ at three
  points in time for scenario \textsc{ahead}: (i) shortly after the initialization,
  (ii) around \SI{100}{\nano\second}, and (iii) after \SI{175}{\nano\second}.
The skews on the three nodes' edges are depicted in
  \cref{fig:sim_skews_ahead}.

\begin{figure}
  \centering
  \includegraphics[width=\linewidth]{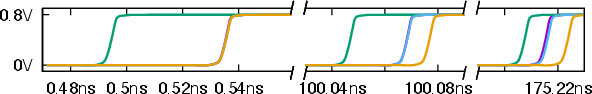}
  \caption{Excerpt of scenario \textsc{ahead}.
  Clock signals of node~$0$ (purple), $1$ (green), $2$ (blue), and $3$ (yellow). We show voltage of the $\clk$ signals on the x-axis over simulated time on the y-axis.
  Nodes from left to right: (i) $1$ before $0$, $2$, $3$, (ii) $1$ before $0$, $2$ before $3$,
    (iii) $1$ before $0$, $0$ slightly before $2$, $2$ before $3$.}
  \label{fig:sim_clocks}
\end{figure}

For the mode signals, in the first scenario, we observe the following:
Since node $1$ is ahead of nodes $0$ and $2$, node $1$'s mode signal is correctly set to
  $0$ (slow mode) while node $0$ and $2$'s mode signals are set to $1$ (fast mode).
Node $3$ is unaware that node $1$ is ahead since it only monitors node $2$.
By default, its mode signal is set to slow mode. Node $2$ then advances its clock
  faster than node $3$.
When the gap between $2$ and $3$ is large enough, node $3$ switches to fast mode.
This configuration remains until nodes $0$ and $2$ catch up to node $1$, where they
  switch to slow mode not to overtake node $1$.
Again, node $3$ sees only node $2$, which is still ahead, and switches to slow mode only after
  it catches up to $2$.

\myparagraph{One Node Behind}
The skews on the edges $(0,1)$, $(1,2)$, and $(2,3)$ are depicted in
  \cref{fig:sim_skews_behind}.
We plot the absolute value of the skew, \eg, at roughly
  \SI{500}{\nano\second} node $1$ overtakes node $0$.
The simulation shows that the algorithm immediately reduces the local skew.
After the system reaches a small local skew after \SI{200}{\nano\second},
  nodes drift relative to each other, \eg, node $2$ drifts ahead of
  node $3$, and node $1$ overtakes node $0$.
The local skew remains in the stable (oscillatory) state after
  \SI{200}{\nano\second} and does not increase significantly.

\myparagraph{Gradient Skew}
The scenario \textsc{gradient} demonstrates how the $\OGCS$ algorithm works.
It reduces the local skew in steps of (odd multiples of) $\kappa$, as
  seen in the plot in \cref{fig:sim_skews_grad} that looks like a staircase.
The algorithm reduces skew on one edge at a time until it reaches the
  next plateau.

\begin{figure}
  \centering
  \includegraphics[width=0.6\linewidth]{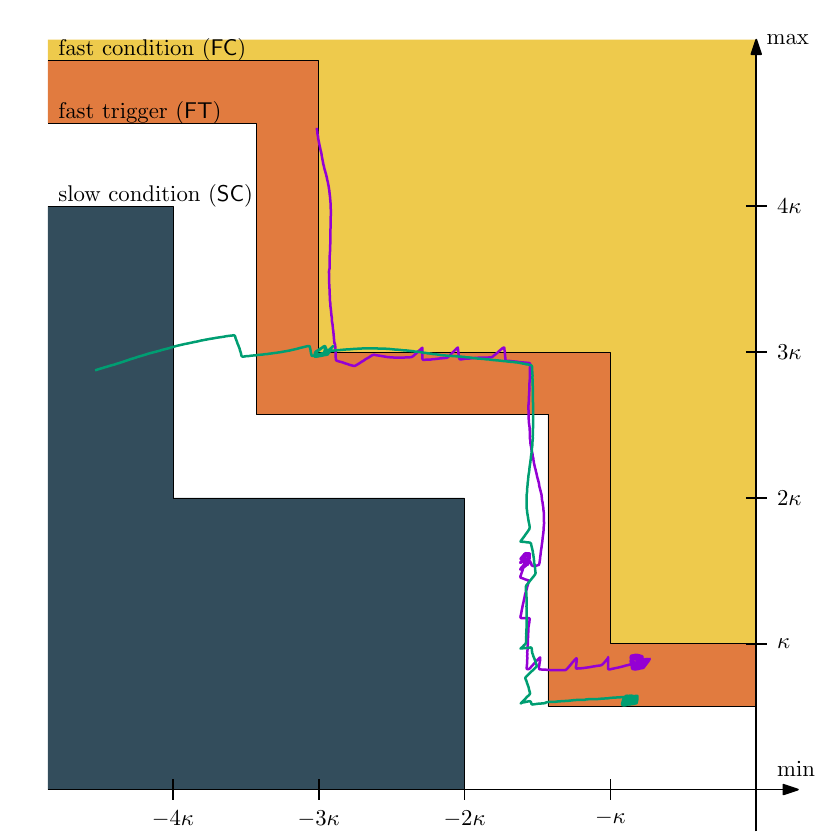}
  \caption{Simulated trajectory of $\widehat{O}_{\max}$ and $\widehat{O}_{\min}$ of node $1$ (purple) and node $2$ (green) from scenario \textsc{gradient} plotted in \cref{fig:gcs_large}}
  \label{fig:gcs_traj_simu}
\end{figure}

\Cref{fig:sim_skews_grad} demonstrates how skew is
  removed.
$\OGCS$ starts by reducing skew on edge $(1,2)$ until it reaches the
  plateau of $(0,1)$ and $(2,3)$.
One by one it then reduces skew on edges $(0,1)$, $(1,2)$ to $(2,3)$ until they reach the
  next plateau.
Finally, it reduces the skews one by one (in reverse order)
  down to a stable range.
\Cref{fig:gcs_traj_simu} shows the same trace in a plot similar to \cref{fig:gcs_large}. 
  
\begin{remark}
  For the sake of a clearly visible convergence of the skew over time,
    we chose initial skews that are beyond the bound in \eqref{eqn:c1}.  
  Thus, the weaker self-stabilizing bound from Theorem~\ref{thm:gcs} applies.
  All simulations showed convergence to a small local skew despite the less
    conservative initialization -- demonstrating the robustness of our algorithm
    also to larger initial skews.
  For our stronger bounds to hold, the reader may consider only
    a respective postfix of the simulation.
\end{remark}

\subsection{Process Variations}
In order to show resilience of our implementation towards process variations,
  we introduce variations to the extracted netlist and run further \textsc{spice} 
  simulations.
We initialize node 1 with a small offset to nodes 0, 2, and 3.
Variations affect the width and length of n-channel and p-channel transistors and the supply
  voltage, where all parameters are simulated at 90\%, 100\%, and 110\% of 
  their typical value.
The resulting skew on edge (0,1) for each simulation is depicted in \cref{fig:variations}.
The simulations show that our system performs well even under process variations.
The local skew after stabilization is below the
  theoretical bound of \SI{20}{\nano\second}.

\begin{figure}
  \centering
  \includegraphics[width=\linewidth]{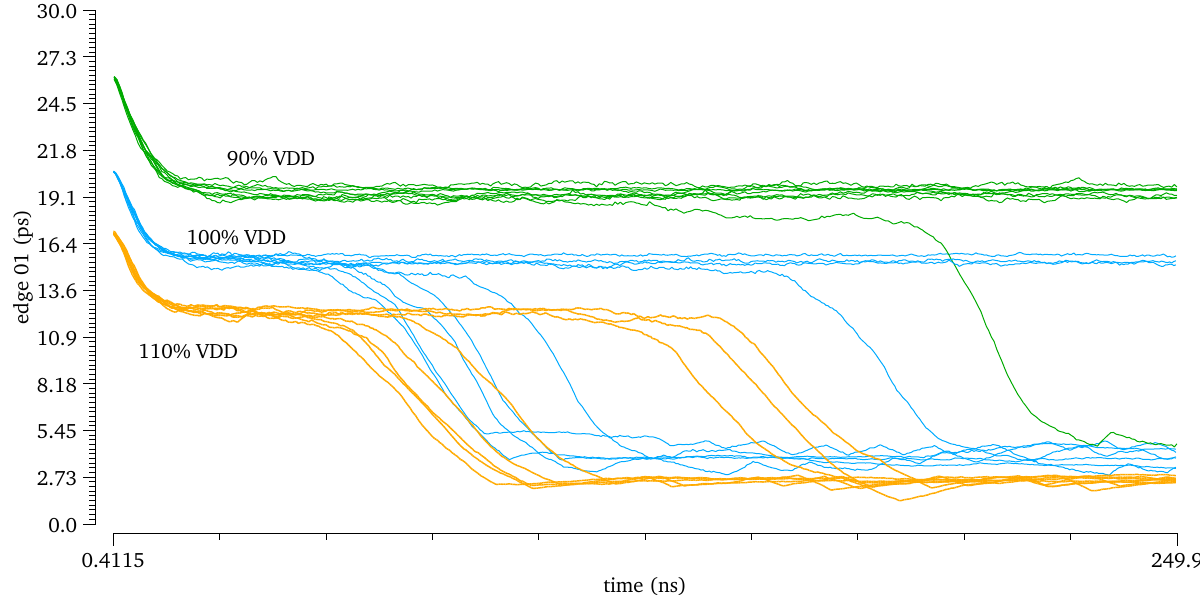}
  \caption{Simulation results for skew (y-axis) over simulated time (x-axis) for edge $(0,1)$ under $90\%$, $100\%$, and $110\%$ variation of the supply voltage and transistor sizes. }
  \label{fig:variations}
\end{figure}

\subsection{Comparison to a Clock Tree}
\label{subsec:clocktree}
For comparison, we laid out a grid of $W \times W$ flip-flops,
  evenly spread in \SI{200}{\micro\metre} distance in x and y direction across the
  chip.
The data port of a flip-flop is driven by the $\OR$ of the up to four
  adjacent flip-flops.
Clock trees were synthesized and routed with Cadence Innovus, with the target
  to minimize skews.
Parameters for delay variations on gates and nets were set to $\pm 5\%$.

For a $2\times2$ grid, Innovus reported an area of $\SI{1.48}{\micro\metre^2}$
  and power of $\SI{0.231}{\milli\watt}$ for the clock tree.
Reported numbers for the PALS system, which comprises $4$ nodes in a line, are
  $\SI{67.09}{\micro\metre^2}$ area and $\SI{0.023}{\milli\watt}$
  power.
Numbers for the PALS system do not include the $4$ starved inverter ring oscillators.
We point out that the size of $\SI{67.09}{\micro\metre^2}$ covers only $0.04\%$
  of the floorplan.
The PALS system uses $171$ gates from the standard cell library.

The resulting clock skews are presented in \cref{fig:clk1}.
We plotted skews guaranteed by our algorithm
  for the same grids with parameters extracted from the
  implementation described in \cref{sec:impl}.
Observe the linear growth of the local clock skew measured in the simulation compared to the
  logarithmic growth of the analytical upper bound on the local skew in our implementation.
The figure also shows the simulated skew for a clock tree with delay variations of $\pm 10\%$.
This comparison is relevant, as $\delta_0$ is governed by \emph{local} delay variations,
  which can be expected to be smaller than those across a large chip.

\begin{figure}
\centering
  \includegraphics[width=.8\linewidth]{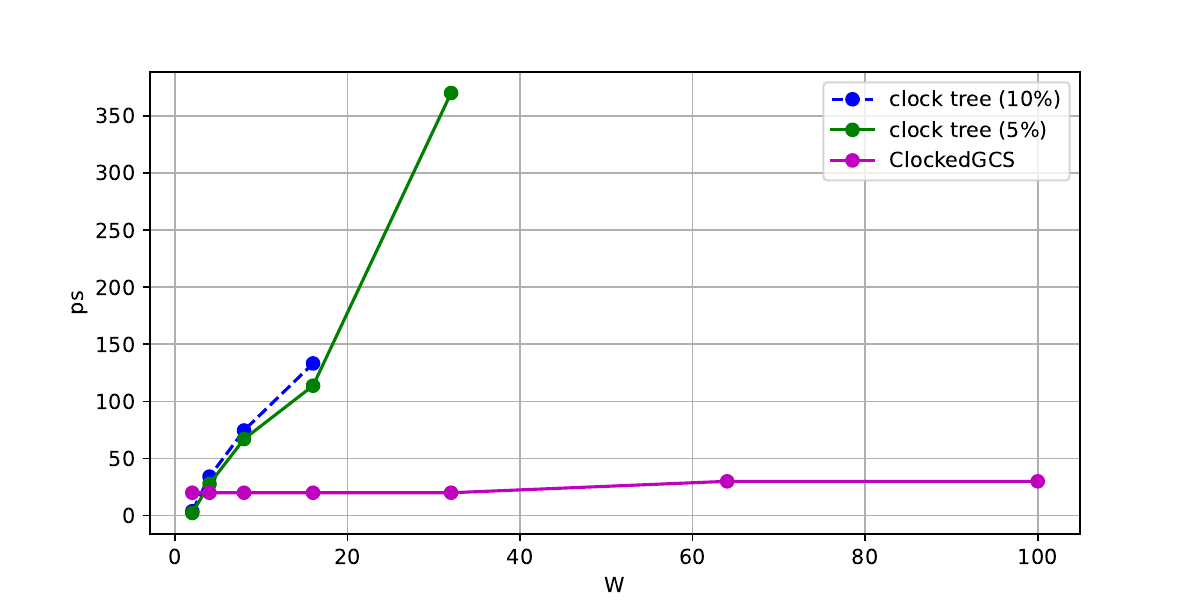}
  \caption{
  Local skew (ps) between neighboring flip-flops in the $W \times W$ grid, with unit length \SI{200}{\micro\metre}.
  Clock tree with $\pm 5\%$ delay variation (solid green) and our algorithm
  with $\pm 5\%$ delay variation (solid magenta).
  The dotted line shows the clock tree with $\pm 10\%$ delay variation,
    demonstrating linear growth of the skew also in a different setting.
  Clock trees are shown up to $W=32$ (i.e., floorplan of width and length $\SI{6.4}{\milli\metre}$) after which Innovus ran out of memory.}
  \label{fig:clk1}
\end{figure}

While a priori the observed linear local skew may be due to the tool, it has been shown
  that no tool can obtain a local skew less than proportional
  to $W$~\cite{fisher85synchronizing}.
This follows from the fact that there are always two
  neighboring nodes in the grid which are in distance proportional to $W$ from
  each other in the clock
  tree~\cite{fisher85synchronizing,boksberger2003approximation}.
Accordingly, uncertainties accumulate in the worst-case fashion to create a local
  skew which is proportional to $W$.
Our algorithm, on the other hand, manages to reduce the local skew exponentially
to being proportional to $\log W$.

To gain intuition on this result, note that there is always an edge that, if
removed (see the edge which is marked by an X in \cref{fig:clktree}),
partitions the tree into two subtrees each spanning an area of $\Omega(W^2)$ and
hence having a shared perimeter of length $\Omega(W)$. Thus, there must be two
adjacent nodes, one on each side of the perimeter, at distance $\Omega(W)$ in
the tree.

\begin{figure}
  \centering
  \includegraphics[width=0.8\linewidth]{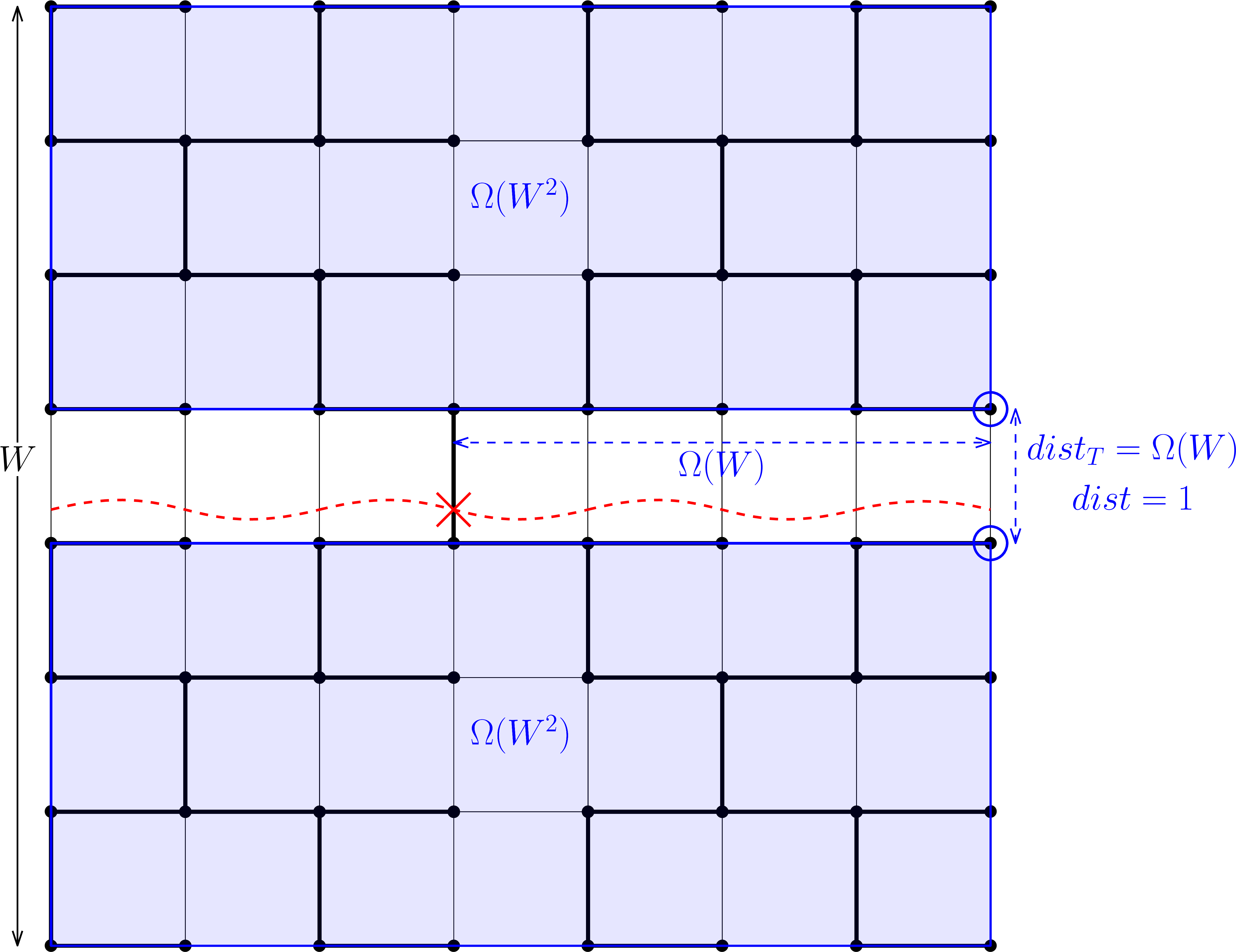}
  \caption{A low stretch spanning tree of an $W\times W$ ($W=8$) grid~\cite{blog}.
    The bold lines depict the spanning tree, \ie, our clock tree in this example.
    The two neighboring nodes that are of distance $13$ in the tree are circled
    (at the middle right side of the grid).}
  \label{fig:clktree}
\end{figure}

\subsection{Comparison to Distributed Clock Generation}
\label{sec:fairbanks}

We next compare our findings to distributed clock generation
  schemes.
Natural are \emph{wait-for-all}
  and \emph{wait-for-one}, where a node produces its next clock tick once it receives
  a tick by one or all its neighbors.
Both approaches are vulnerable to large local skews, however.

A clever combination of both schemes is used by the clock generation
  grid by Fairbanks and Moore~\cite{fairbanks2005self}.
In the grid, both approaches alternate for adjacent nodes.
For comparison, we simulated a digital abstraction of the clock generation
  grid.
Based on ideas of the lower bound proof for local skews~\cite{fan2006gradient},
  we construct a simulation scenario that demonstrates that large local skew
  are possible in the clock grid, however.

\myparagraph{Clock Generation Grid}
The clock generation grid is a self-timed analog circuit that provides
  local, synchronized clocks.
It is based on the \emph{Dynamic asP} \textsc{fifo}
  control by Molnar and Fairbanks~\cite{molnar1999control}.
While the optimized version of the clock generation grid is an analog  implementation that involves rigorous transistor sizing and layout,
  we focus on a digital version~\cite{fairbanks2005self}
  which is easier to adapt and manufacture in a standard design process.
We distinguish two types of nodes: pull-up nodes and pull-down nodes
  (see \cref{fig:fairbanks-implementation}).
On every edge between two nodes, there is a set-reset latch.
Pull-up nodes set the latch and pull-down nodes reset the latch.
Pull-up nodes compute the logical $\NOR$ of incoming edges
  (equivalent to a wait-for-all approach) and set the latch.
Pull-down nodes compute the logical $\AND$ of incoming edges
  (equivalent to a wait-for-one approach) and reset the latch.
The clock of a node is derived by the output of the respective $\NOR$
  or $\AND$.
The grid's frequency is easily adjusted by adding a delay
  between the nodes and the latches.

\begin{figure}
  \centering
  \includegraphics[width=0.7\linewidth]{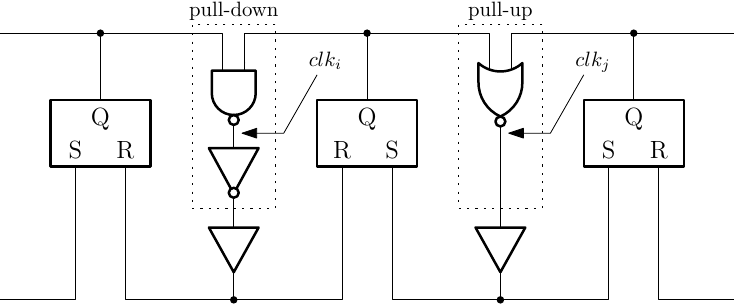}
  \caption{Schematic of the digital abstraction of the Fairbanks clock generation on a line showing one pull-down and one pull-up node.}
  \label{fig:fairbanks-implementation}
\end{figure}

\myparagraph{Setup}
We next conducted simulations that examine the behavior of different
  communication delays.
In order to simulate slower communication paths, we add a small
  capacity (\SI{0.01}{\pico\farad}) to the communication channel.
For simplicity, we differentiate between two delays: fast and slow.

Formally, the algorithm combining wait-for-all and wait-for-one approaches
  can have a local skew that grows linearly with the network diameter.
Through our simulations, we demonstrate that this is indeed possible in the clock generation
  grid and that $\OGCS$ can cope with this situation.
The setup of the simulated delays is depicted in \cref{fig:fairbanks-setup}:
  outgoing edges of nodes $2$, $3$, and $4$ are fast and edges outgoing from $0$, $1$, $5$, and $6$
  are slow.

\begin{figure}
  \includegraphics[width=\linewidth]{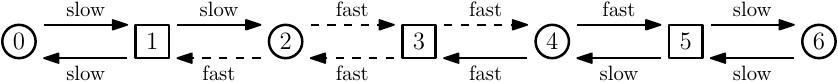}
  \caption{Delay setup with fast and slow message delays for the Fairbanks clock generation on a line. The setup achieves a large local skew. Round nodes denote pull-down nodes
  and rectangular nodes denote pull-up nodes. In the lower bound simulation dashed edges are
  swapped from fast to slow communication.}
  \label{fig:fairbanks-setup}
\end{figure}

\myparagraph{Results for Clock Generation Grid}
The digital clock generation grid runs at a frequency of
  about \SI{2.5}{\giga\hertz}.
By simulations, we determined that an additional capacity of \SI{0.01}{\pico\farad} on an edge
  adds a delay of approximately \SI{7}{\pico\second}.
We next conducted simulations with three different delay settings:
\textsc{nocap}: the grid without additional delays,
\textsc{fullcap}: the grid with the capacity added to every edge,
and \textsc{large-local}: the grid with the setting from \cref{fig:fairbanks-setup}.

Local skews of simulation scenarios \textsc{nocap}, \textsc{fullcap},
  and \textsc{large-local} are shown in \cref{fig:localskews}.
We observe that the grid achieves a low skew if the delay is uniform
  on all edges.
For scenario \textsc{nocap} (\resp \textsc{fullcap}) we measure a local
  skew of \SI{15}{\pico\second} (\resp \SI{16}{\pico\second})
  and global skew of \SI{22}{\pico\second} (\resp \SI{23}{\pico\second}).
By contrast, the grid experiences poor synchronization
  for non-uniform delays (\textsc{large-local}) where we obtained
  a local skew of \SI{38}{\pico\second} and a global skew of \SI{61}{\pico\second}.

\begin{figure}
  \begin{subfigure}{.5\linewidth}
    \centering
    \includegraphics[width=\linewidth,clip,trim=65 0 61 0]{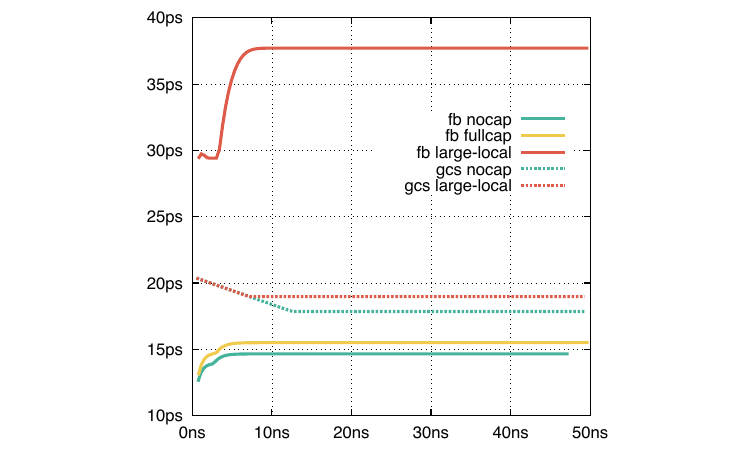}
    \caption{Fairbanks (fb) and GCS}
    \label{fig:localskews}
  \end{subfigure}
  \begin{subfigure}{.5\linewidth}
    \centering
    \includegraphics[width=\linewidth,clip,trim=65 0 61 0]{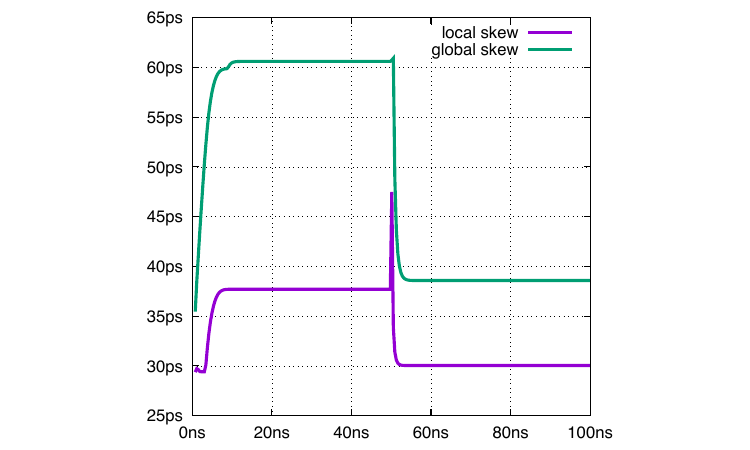}
    \caption{Fairbanks Line}
    \label{fig:fbswap}
  \end{subfigure}
  \caption{Comparison of Fairbanks to $\CGCS$. (a) Local skews of the experiments \textsc{nocap}, \textsc{fullcap}, and \textsc{large-local}. (b) Local and global skew of the lower bound simulation.}
\end{figure}

\myparagraph{Lower Bound Simulation}
In this simulation, we apply ideas from the formal argument for
  lower bounds on wait-for-all and wait-for-one approaches.
The idea is to build up a large global skew and then 
  change the edges' delays step by step, pushing the global skew onto a single edge.
The simulation in \cref{fig:fbswap} shows one of these push-steps.
In the first part of the simulation (until \SI{50}{\nano\second})
  the system builds up a large local skew.
At \SI{50}{\nano\second} we switch delays of edges outgoing from nodes $3$ and $4$
  as described in \cref{fig:fairbanks-setup}.
As expected, the global skew is pushed onto the local skew right after
 (we measure \SI{47}{\pico\second}).
Following the argument of the lower-bound proof, one can repeat the procedure to
  push the complete global skew (temporarily) onto a single edge.

\myparagraph{Comparison to $\BFM$}
Our clock generation algorithm runs at a frequency of
  about \SI{2}{\giga\hertz}.
By simulations, we measured that the added capacity of
   \SI{0.01}{\pico\farad} adds a delay of approximately \SI{1}{\pico\second}.
We observed a local skew of \SI{18}{\pico\second} and a global skew of \SI{20}{\pico\second}
  in the absence of additional communication delay.
This is slightly worse than the local skew of the grid (see \cref{fig:localskews}).
However, by our theoretical findings, our algorithm does not suffer from a large skew in the setting
  where delays in the center are fast or switched from fast to slow:
  simulations of \cref{fig:fairbanks-setup} showed a local skew of
  \SI{19}{\pico\second} and a global skew of \SI{20}{\pico\second}.

\section{Conclusion}\label{sec:concl}
The presented $\CGCS$ algorithm is a clock synchronization algorithm that provably
  maintains the skew bounds of the \textsf{GCS} algorithm by Lenzen et al.~\cite{lenzen10tight}. 
The algorithm can be deconstructed into parts that hardware modules can implement.
Asymptotically, the algorithm maintains a local skew that is at most logarithmic in the chip's width, whereas clock trees are shown to perform only linear in the width of the chip.

By simulation, we show that our $\BFM$ implementation of these modules in a $\SI{15}{\nano\metre}$ FinFET
  achieves small skews between neighbors even under process variations.
We discuss different simulation setups that show how the algorithm behaves. 

In future work, we aim to improve the analysis of the \textsf{GCS} algorithm by showing its correctness for oscillators with less stringent requirements. Additionally, we plan to conduct simulations of a full design, incorporating PVT and PPA experiments, and eventually produce an ASIC chip for comparison to a cutting-edge GALS design.

\ifCLASSOPTIONcompsoc
  \section*{Acknowledgments}
\else
  \section*{Acknowledgment}
\fi
The authors would like to thank the team of EnICS Labs 
for support and helpful discussions.
In particular, we thank Benjamin Zambrano and Itamar Levi, Itay Merlin, Shawn Ruby, Adam Teman, Leonid Yavits. 
This project has received funding from the
European Research Council (ERC) under the European
Union’s Horizon 2020 research and innovation programme
(grant agreement 716562). 
This research
was supported by the Israel Science Foundation under Grant 867/19 
and by the ANR project DREAMY (ANR-21-CE48-0003).
\ifCLASSOPTIONcaptionsoff
  \newpage
\fi
\bibliographystyle{IEEEtran}
\bibliography{biblio.bib}
\begin{IEEEbiography}
[{\includegraphics[width=1in,height=1.25in,clip,keepaspectratio,trim=5 0 4 0]{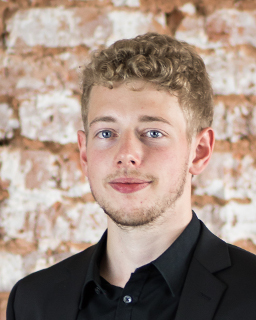}}]{Johannes~Bund}
is a post-doc researcher at the Faculty of Engineering at Bar-Ilan University since 2022.
He graduated with his M.\,Sc.\ studies in 2018 at the Saarland Informatics Campus and
  Max-Planck Institute for Informatics.
In 2018 he joined Christoph Lenzen's group at Max-Planck Institute for Informatics as 
  a Ph.\,D.\ student.
In 2021 he switched, together with Christoph Lenzen, to CISPA Helmholtz Center for Information
  Security, where he finished his Ph.\,D.\ studies.
\end{IEEEbiography}
\begin{IEEEbiography}
[{\includegraphics[width=1in,height=1.25in,clip,keepaspectratio]{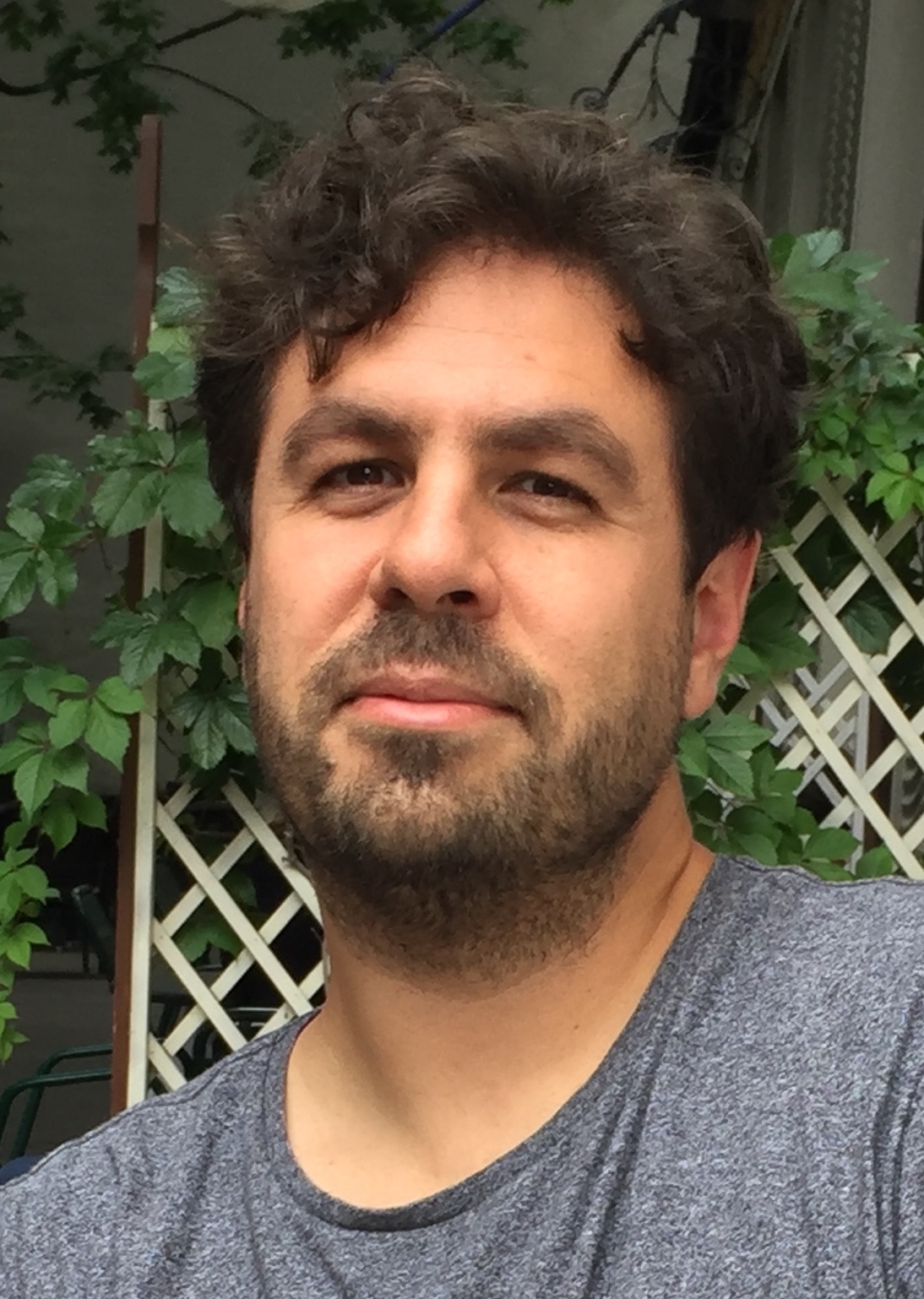}}]{Matthias~F\"ugger}
received his M.\,Sc.\ (2006), and his Ph.\,D. (2010) in computer engineering from TU Wien, Austria.
He worked as an assistant professor at TU Wien and as a post-doctoral researcher at LIX, Ecole Polytechnique, and at MPI for Informatics.
Currently, he is a CNRS researcher at LMF, ENS Paris-Saclay, where he leads the
  Distributed Computing group.
\end{IEEEbiography}
\begin{IEEEbiography}[{\includegraphics[width=1in,height=1.25in,clip,keepaspectratio]{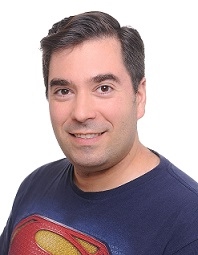}}]{Moti Medina}
is a faculty member in the engineering faculty at Bar-Ilan University since 2021.
Previously he was a faculty member at the School of Electrical \& Computer Engineering at
the Ben-Gurion University of the Negev since 2017. Previously, he was a post-doc
researcher in MPI for Informatics and in the Algorithms and Complexity group at
LIAFA (Paris 7). He graduated with his Ph.\,D., M.\,Sc., and B.\,Sc.\ studies at the
School of Electrical Engineering at Tel-Aviv University, in  2014, 2009, and 2007
respectively. Moti is also a co-author of a  text-book on logic design
``Digital Logic Design: A Rigorous Approach'', Cambridge Univ. Press, 2012.
\end{IEEEbiography}
\vfill
\end{document}